\documentclass[pagebackref]{jocg}
\usepackage{amsmath,amsfonts}
\usepackage{microtype}
\usepackage{xspace}
\usepackage{wrapfig}
\usepackage{graphicx}
\usepackage{amssymb}
\usepackage{xcolor}

\usepackage{array}

\graphicspath{{./figures/}}

\title{%
  \MakeUppercase{A Framework for Algorithm Stability\newline and its Application to Kinetic Euclidean MSTs}
	\thanks{A preliminary version of this work~\cite{meulemans2018framework} was presented at the 13th Latin American Symposium on Theoretical Informatics (LATIN 2018) and this work is part of the last author's PhD thesis~\cite{julesthesis}.
	}
}

\author{%
  Wouter Meulemans,%
  \thanks{\affil{TU Eindhoven}, 
          \email{\{w.meulemans,b.speckmann,k.a.b.verbeek,j.j.h.m.wulms\}@tue.nl}}\,
  Bettina Speckmann,%
  \footnotemark[2]\,
  Kevin Verbeek,%
  \footnotemark[2]\,
  and Jules Wulms\footnotemark[2]
}

\usepackage{amsthm}
\theoremstyle{plain}
\newtheorem{theorem}{Theorem}
\newtheorem{lemma}{Lemma}
\newtheorem{corollary}{Corollary}

\newcommand{\reals}{\mathbb{R}}
\newcommand{\prop}[1]{\mathrm{\textsc{#1}}}
\DeclareMathOperator{\OPT}{OPT}
\DeclareMathOperator{\St}{St}
\DeclareMathOperator{\TS}{\rho_{TS}}
\DeclareMathOperator{\LS}{\rho_{LS}}
\newcommand{\etal}{\emph{et al.}\xspace}

\newcommand{\myparNS}[1]{\noindent{\bfseries #1.}}
\newcommand{\mypar}[1]{\medskip\myparNS{#1}}

\begin{document}
\maketitle

\begin{abstract}
We say that an algorithm is \emph{stable} if small changes in the input result in small changes in the output. This kind of algorithm stability is particularly relevant when analyzing and visualizing time-varying data. Stability in general plays an important role in a wide variety of areas, such as numerical analysis, machine learning, and topology, but is poorly understood in the context of (combinatorial) algorithms.
In this paper we present a framework for analyzing the stability of algorithms. We focus in particular on the trade-off between the stability of an algorithm and the quality of the solution it computes.
Our framework allows for three types of stability analysis with increasing degrees of complexity: event stability, topological stability, and Lipschitz stability. In addition, we need to refine the model of an algorithm based on how it interacts with the time-varying data, for which we consider several options.
We demonstrate the use of our stability framework by applying it to kinetic Euclidean minimum spanning trees.
\end{abstract}

\section{Introduction}\label{sec:intro}

With recent advances in sensing technology, vast amounts of \emph{time-varying data} are generated, processed, and analyzed on a daily basis.
Hence there is a great need for algorithms that can operate efficiently on time-varying data and that can offer guarantees on the quality of analysis results. A specific relevant subset of time-varying data consists of so-called \emph{motion data}: geolocated, and hence geometric, time-varying data. To deal with the challenges of motion data, Basch~\etal~\cite{basch1999data} in 1999 introduced the \emph{kinetic data structures} (KDS) framework. Kinetic data structures efficiently maintain a (combinatorial) structure on a set of moving objects. The KDS framework has sparked a significant amount of research, resulting in many efficient algorithms for motion data.

The performance of a particular algorithm is usually judged with respect to a variety of criteria, with the two most common criteria being solution quality and running time. These criteria are conflicting: algorithms that find optimal solutions are typically slower than algorithms that are more lenient in solution quality, and compute only approximations to optimal solutions. In the context of algorithms for time-varying data, a third important criterion is \emph{stability}. Whenever analysis results need to be communicated to humans, for example via visual representations, it is important that these results are \emph{stable}: small changes in the data result in small changes in the output. 

These changes in the output can be continuous or discrete. As an example, consider maintaining a Euclidean minimum spanning tree (EMST) of a set of moving points. As the points move, edges continuously grow and shrink in length. However, at certain points in time, one edge has to be swapped for another edge to maintain minimum length, and hence no algorithm that maintains an optimal EMST can be stable. If the output of our algorithm is a visualization, then sudden changes in the visual representation of data disrupt the so-called \emph{mental map}~\cite{kitchin1994cognitive} of the user and prevent the recognition of temporal patterns. Stability also plays a role if changes in the result of an algorithm are costly in practice (e.g. in physical network design, where frequent significant changes to the network are expensive). Attaining stability is non-trivial, though, since this criterion brings new trade-offs into play (see Figure~\ref{fig:trade-offs}): for example, an optimal solution may not be stable, and hence computing a stable solution requires an algorithm to deviate from the optimum.

\begin{figure}
	\centering
	\includegraphics[page=1]{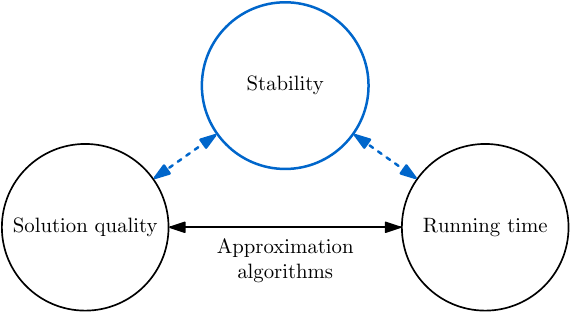}
	\caption{Trade-offs between quality criteria for algorithms. Traditional criteria and their trade-off are shown in black; additional trade-offs for time-varying data are dashed blue.}
	\label{fig:trade-offs}
\end{figure}

The stability of algorithms and other (computational) methods has been well-studied in a variety of research areas, such as numerical analysis~\cite{higham2002accuracy}, machine learning~\cite{bousquet2002}, control systems~\cite{bacciotti2006liapunov}, and topology~\cite{cohen2005stability}. In contrast,
the stability of combinatorial algorithms for time-varying data has received little attention in the theoretical computer science community so far. Here, it is of particular interest to understand the trade-offs between solution quality, running time, and stability. To illustrate the trade-off between solution quality and stability, again consider an EMST. As the points move, the edges may have to change frequently and significantly to maintain minimum length. On the other hand, if we start with an EMST for the input point set and never change it combinatorially, then the spanning tree we maintain is very stable -- but over time it can devolve into a very long spanning tree.

Our goal, and the focus of this paper, is to understand the possible trade-offs between solution quality and stability. This is in contrast to earlier work on stability in other research areas, such as the ones mentioned above, where stability is usually considered in isolation. Since there are currently no suitable tools available to formally analyze trade-offs involving stability, we introduce a new analysis framework. We believe that there are many interesting and relevant questions to be solved in the general area of algorithm-stability analysis and our framework is a first meaningful step towards tackling them.

\mypar{Results and organization} We present a framework to analyze the stability of (combinatorial) algorithms.
In this paper, we focus on analyzing the trade-off between stability and solution quality, omitting running time from consideration.
Our framework allows for three types of stability analysis of increasing degrees of complexity, along with increasing stability requirements: \emph{event stability}, \emph{topological stability}, and \emph{Lipschitz stability}.
These stability types can be applied both to motion data and to more general time-varying data. Besides these modes of analysis, we distinguish between three models for algorithms on time-varying data: \emph{stateless}, \emph{state-aware}, and \emph{clairvoyant} algorithms. The model of an algorithm is based on the availability of the time-varying input and influences how much stability can be achieved.
We demonstrate the use of our stability framework by applying it to the problem of kinetic Euclidean minimum spanning trees (EMSTs). Some of our results for kinetic EMSTs are directly applicable in other settings.

In Section~\ref{sec:framework} we give an overview of our framework for the analysis of algorithm stability. In Sections~\ref{sec:eventstable}, \ref{sec:topostable}, and \ref{sec:Lipstable} we describe event stability, topological stability, and Lipschitz stability, respectively. In each of these sections we first describe the stability analysis in a generic setting and then use it to analyze the kinetic EMST problem.In Section~\ref{sec:conclusion} we close with some concluding remarks on our stability framework.

\mypar{Related work} Stability is a natural concern in visually oriented research areas such as graph drawing, (geo-)visualization, and automated cartography. For example, in dynamic map labeling~\cite{been2010optimizing,gemsa2011sliding,gemsa2016consistent,nollenburg2010dynamic},
the \emph{consistent dynamic labeling} model allows a label to appear and disappear only once, making it very stable. There are very few theoretical results, with the
noteworthy exception of so-called simultaneous embeddings~\cite{brass2007simultaneous,frati2009constrained} in graph drawing, which can be seen as a very restricted model of stability. However, none of these results offer any real structural insight into the trade-off between solution quality and stability.

In combinatorial optimization, Bilu and Linial~\cite{DBLP:journals/cpc/BiluL12} introduced the notion of $\gamma$-stable solutions, which describes the trade-off between solution quality and stability for instances of discrete optimization problems: an instance is $\gamma$-stable if its unique optimal solution remains unchanged after a perturbation with a factor~$\gamma$. In computational geometry there are also a few results on the trade-off between solution quality and stability. Specifically, the stability of kinetic 1-center and 1-median problems has been studied by Bespamyatnikh~\etal~\cite{DBLP:conf/dialm/BespamyatnikhBKS00}. They developed approximations by fixing the speed at which the center/median point moves and provide some results on the trade-off between solution quality and speed. Durocher and Kirkpatrick~\cite{durocher2006steiner} study the stability of centers of kinetic point sets as well, and define the notion of $\kappa$-stable center functions, which is closely related to our concept of Lipschitz stability. In later work \cite{durocher2008bounded} they consider the trade-off between the solution quality of Euclidean $2$-centers and a bound on the velocity with which the $2$-centers can move. De Berg~\etal~\cite{de2013kinetic} show similar results in the black-box KDS model.

One can argue that the KDS framework~\cite{guibas2004kinetic} already indirectly considers stability in a limited form, namely as the number of \emph{external events}, which are (combinatorial) changes in the maintained structure. However, internally, a KDS often processes more events than just the external events. For efficiency, a KDS should process as few events as possible, with the number of external events being a fixed lower bound, since a KDS may not deviate from the chosen structure. In contrast, we allow ourselves to sacrifice solution quality to compute a more stable approximate structure with fewer external events. 

Kinetic Euclidean minimum spanning trees and related structures have been studied extensively.
Katoh~\etal~\cite{katoh1995minimum} proved an upper bound of $O(n^3 2^{\alpha(n)})$ for the number of external events of EMSTs of $n$ linearly moving points, where $\alpha(n)$ is the inverse Ackermann function.
Rahmati~\etal~\cite{rahmati2015simple} present a kinetic data structure for EMSTs in the plane that processes $O(n^3 \beta_{2s + 2}^2(n) \log n)$ events in the worst case, where $s$ is a measure for the complexity of the point trajectories and $\beta_s(n)$ is an extremely slow-growing function.
The best known lower bound for external events of EMSTs in $d$ dimensions is $\Omega(n^d)$~\cite{monma1992transitions}. 
Since the EMST is a subset of the Delaunay triangulation, one can also consider to kinetically maintain the Delaunay triangulation instead. Fu and Lee~\cite{fu1991voronoi}, and Guibas~\etal~\cite{guibas1992voronoi} show that the Delaunay triangulation undergoes $O(n^2 \lambda_{s+2}(n))$ external events (near-cubic), where $\lambda_s(n)$ is the maximum length of an $(n, s)$-Davenport-Schinzel sequence~\cite{sharir1995davenport}. On the other hand, the best known lower bound for external events of the Delaunay triangulation is only $\Omega(n^2)$~\cite{sharir1995davenport}. Rubin improves the upper bound to $O(n^{2+\varepsilon})$, for any $\varepsilon > 0$, if the number of degenerate events is limited~\cite{rubin2013topological}, or if the points move along a straight line with unit speed~\cite{rubin2015kinetic}). Agarwal~\etal~\cite{agarwal2015stable} also consider a more stable version of the Delaunay triangulation, which undergoes at most a nearly quadratic number of external events. However, external events for EMSTs do not necessarily coincide with external events of the Delaunay triangulation~\cite{rahmati2012kinetic}. We can also consider approximations of the EMST, for example via spanners or well-separated pair decompositions~\cite{arya2016fast}. However, kinetic $t$-spanners already undergo $\Omega(\frac{n^2}{t^2})$ external events~\cite{gao2006deformable}. Our stability framework allows us to reduce the number of external events even further and to still state something meaningful about the quality of the resulting EMSTs.

Since the first introduction of our framework in 2018~\cite{meulemans2018framework}, multiple notions of stability have been proposed for various optimization problems. Akitaya~\etal~\cite{DBLP:journals/jgaa/AkitayaBBCMSS23} studied computing a single spanning tree for a set of moving points. They consider different optimization functions that lead to NP-hardness results as well as (approximation) algorithms. In the dynamic setting, De Berg~\etal introduced $k$-stable algorithms~\cite{DBLP:conf/walcom/BergMS24,DBLP:conf/approx/BergSS23,DBLP:journals/siamdm/BergSS24}, which make up to $k$ changes to the solution at every dynamic insertion and/or deletion in the input. Furthermore, Kumabe and Yushida consider Lipschitz continuous algorithms for graph and covering problems, as well as optimization games~\cite{DBLP:journals/corr/abs-2307-08213,DBLP:conf/focs/KumabeY23,DBLP:conf/icalp/KumabeY24}. Additionally, Abdelkader and Mount~\cite{abdelkader2025differentiable} studied the trade-off between solution quality and differentiability, imposing some kind of continuity on their solution space: they compute a data structure for approximate distance queries, which can additionally return the gradient of the distance function at any query point, to facilitate gradient-based optimization. Finally, our framework has been used to study the stability of the $k$-center problem~\cite{hoog2018topological} and of shape descriptors, such as bounding boxes and bounding strips~\cite{meulemans2019stability}.

\section{Stability framework}\label{sec:framework}

Intuitively, we can say that an algorithm is stable, if small changes in the input lead to small changes in the output. More formally, let $\Pi$ be an optimization problem that, given an input instance $I$ from a set $\mathcal{I}$, asks for a solution $S$ from a set $\mathcal{S}$ that minimizes (or maximizes) some optimization function $f\colon \mathcal{I}\times\mathcal{S} \rightarrow \reals$.

An algorithm $\mathcal{A}$ for $\Pi$ maps an input instance $I$ to a solution $S$. For an optimal solution $\OPT\colon \mathcal{I} \rightarrow \mathcal{S}$ of input instance $I$, it holds that $f(I, \OPT(I)) = \min_{S \in \mathcal{S}} f(I, S)$ is minimal over all possible solutions $S \in \mathcal{S}$ for $I$. Note that we did not define $\mathcal{A}$ as a function between input and output. This is intentional and will be further elaborated on in Section~\ref{subsec:algmodels}, where we distinguish between various algorithmic models that interact differently with the time-varying aspect of the data.

When working with time-varying data, the input instance is not static, and for the purpose of introducing time-varying input, we model time-varying input instances as a sequence of $T$ static inputs $I_1, I_2, \ldots, I_T$. In this time-varying setting, algorithm $\mathcal{A}$ maps the input sequence $I_1, I_2, \ldots, I_T$ to a sequence of solutions $S_1, S_2, \ldots, S_T$. An optimal solution $\OPT$ for $\Pi$ on this time-varying input is then defined by individual static solutions $\OPT$ for each input $I_i$ in the sequence, such that $f(I_i,\OPT(I_i))$ is minimized for each input separately. Note that such a solution has optimal solution quality, and completely ignores stability.

To define the \emph{stability} of an algorithm, we need to quantify changes in the input instances and in the solutions. We can do so by imposing a metric on $\mathcal{I}$ and $\mathcal{S}$. Let $d_{\mathcal{I}}\colon \mathcal{I}\times\mathcal{I} \rightarrow \reals_{\geq 0}$ be a metric for $\mathcal{I}$ and let $d_{\mathcal{S}}\colon \mathcal{S}\times\mathcal{S} \rightarrow \reals_{\geq 0}$ be a metric for $\mathcal{S}$. Figure~\ref{fig:input-solution-space} gives an overview of the above definitions. For an input sequence $\sigma_I = I_1, I_2, \ldots, I_T$ and a sequence of solutions $\sigma_S = S_1, S_2, \ldots, S_T$ we can define the stability of this solution relative to the input as follows.
\begin{equation}\label{eq:discretestability}
\St(\sigma_I, \sigma_S) = \max_{i \in [1,T-1]} \frac{d_\mathcal{S}(S_i,S_{i+1})}{d_\mathcal{I}(I_i,I_{i+1})}
\end{equation}
Finally, let $\Sigma_\mathcal{I}$ be the set of all sequences of input instances. We can then define the stability of an algorithm $\mathcal{A}$, which maps every input $I_i$ in sequence $\sigma_I$ to a solution $S_i$ in a sequence of solutions $\sigma_S$, as follows.
\begin{equation}\label{eq:discretealgorithmstability}
\St(\mathcal{A}) = \sup_{\sigma_I \in \Sigma_I} \St(\sigma_I, \sigma_S)
\end{equation}
The lower the value for $\St(\mathcal{A})$, the more stable we consider the algorithm $\mathcal{A}$ to be. Note that $\St(\mathcal{A})$ does not depend on $f$, and hence the stability analysis is not dependent on whether $\Pi$ is a minimization or maximization problem.

\begin{figure}
	\centering
	\includegraphics{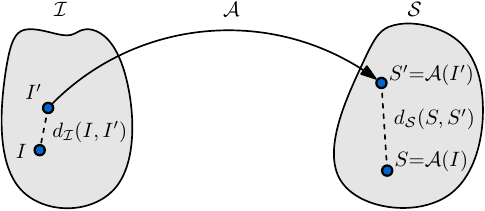}
	\caption{Algorithm $\mathcal{A}$ maps input instances from the input space $\mathcal{I}$ to the solution space $\mathcal{S}$. Metrics $d_\mathcal{I}$ and $d_\mathcal{S}$ allow us to reason about how different two instances are in respectively the input and solution space.}
	\label{fig:input-solution-space}
\end{figure}

There are many other ways to define the stability of an algorithm given the metrics, and the definition above does not cover all scenarios in which we want to analyze stability. For example, we could use algorithms to compute visualizations of time-varying data or the solution of our algorithm could correspond to objects in the physical world that are costly to change or take time to move. In the above definition of stability, if there is no change in the input ($I_i = I_{i+1}$) then there should not be any change in the output either ($S_i = S_{i+1}$), otherwise $\St(\mathcal{A})$ is unbounded. This fits the visualization example, where we want to see no changes in the visualization, unless the data is changing. However, consider a group of security officers at a public event, the officers can (slowly) move into better positions to cover the crowd, even though the crowd may be (mostly) stationary. We can adapt Equation~\ref{eq:discretestability} to express these conditions as well.
\begin{equation}\label{eq:discretestability2}
\St(\sigma_I, \sigma_S) = \frac{\max_{i \in [1,T-1]} d_\mathcal{S}(S_i,S_{i+1})}{\max_{i \in [1,T-1]} d_\mathcal{I}(I_i,I_{i+1})}
\end{equation}
In this definition, instead of using the amount of change between time steps, we use an upper bound for both the input and solution. Since the upper bound on the changes in the solution can occur between different time steps than the upper bound on the changes in the input, $\St(\mathcal{A})$ will not be unbounded when the solution changes without any change in the input. Our framework will allow for stability analysis that can adhere to either of these models, whichever fits the problem at hand.

Kinetic time-varying data is typically gathered through geo-location, as already explained in the introduction. Hence, the above setting would suffice to analyze the stability of this kind of time-varying data. However, as the rate at which the movement data is recorded becomes higher, and the difference between consecutive data points becomes smaller, the data closely resembles the continuous motion of real-life objects. In our theoretical framework we would therefore like to abstract from the discrete nature of the data, and work with continuously changing inputs. This approach is very natural for geometric problems and is used abundantly in computational geometry, for example by using continuously moving points as input.

Instead of a sequence of static inputs, we now get a continuous path $I\colon [0, 1] \rightarrow \mathcal{I}$ through the input space $\mathcal{I}$. For each time $t \in [0,1]$, algorithm $\mathcal{A}$ maps $I(t) \in \mathcal{I}$ to a solution $S(t) \in \mathcal{S}$. Note that the resulting time-varying solution $S\colon [0, 1] \rightarrow \mathcal{S}$ is not necessarily continuous. An optimal solution $\OPT$ for input $I$ minimizes (or maximizes) $f(I(t),\OPT(I(t)))$ for every $t\in [0,1]$, and may have discontinuities as well. The abstraction to (continuous) paths allows us to make a distinction between continuous and discrete changes in the output, which possibly lead to instabilities.

The stability of a solution $S$ relative to $I$ can then be defined similarly to the discrete case, using the metrics on $\mathcal{I}$ and $\mathcal{S}$. The stability is then a generalization of Equation~\ref{eq:discretestability}
\begin{equation}\label{eq:continuousstability}
\St(I, S) = \sup_{t \in [0,1)} \lim_{\varepsilon\to 0} \frac{d_\mathcal{S}(S(t),S(t+\varepsilon))}{d_\mathcal{I}(I(t),I(t+\varepsilon))}
\end{equation}
In this definition of stability, part of the formula is essentially a metric derivative, as we compare the speed along parameterized paths $I(t)$ and $S(t)$ in metric spaces $(\mathcal{I}, d_\mathcal{I})$ and $(\mathcal{S}, d_\mathcal{S})$, respectively. Let $\mathcal{P}_\mathcal{I}$ be the set of continuous paths through $\mathcal{I}$. The stability of an algorithm $\mathcal{A}$ can now also be generalized from Equation~\ref{eq:discretealgorithmstability}. 
\begin{equation}\label{eq:continuousalgorithmstability}
\St(\mathcal{A}) = \sup_{I \in \mathcal{P}_\mathcal{I}} \St(I, S)
\end{equation}
Again, this definition assumes that no change in the input should result in no change in the output, otherwise $\St(\mathcal{A})$ is unbounded. If we want the stability to be less strict as in Equation~\ref{eq:discretestability2}, we can also generalize this to the continuous setting:
\begin{equation}\label{eq:continuousstability2}
\St(I, S) = \frac{\sup_{t \in [0,1)} \lim_{\varepsilon\to 0} d_\mathcal{S}(S(t),S(t+\varepsilon))}{\sup_{t \in [0,1)} \lim_{\varepsilon\to 0} d_\mathcal{I}(I(t),I(t+\varepsilon))}
\end{equation}

\subsection{Stability vs.~solution quality}
\label{subsec:stabilitytypes}
Optimal solutions are not necessarily stable, since small changes in the input can drastically change how an optimal solution looks.
In the above definitions such a scenario leads to instability as well: the stability of optimal solutions $\St(\OPT)$ is influenced by how far apart consecutive solutions are in the solution space, with respect to the distance between the corresponding inputs in the input space. For many optimization problems, the optimal solutions produced by $\OPT$ may be very unstable. The instabilities can be the result of $\OPT$ changing very fast, albeit continuously, or due to $\OPT$ having discontinuities. This suggests an interesting trade-off between the stability of an algorithm and the solution quality: since $\OPT$ is unstable, stable solutions cannot have optimal solution quality either.

Although the definitions described above (in Equations~\ref{eq:discretealgorithmstability} and \ref{eq:continuousalgorithmstability}) nicely capture the stability of an algorithm, they do not directly allow a feasible analysis of the trade-off between stability and solution quality: it is hard to compute the stability of an algorithm (as it depends on all input sequences/paths), let alone reason about all possible algorithms. Furthermore, it requires the definition of suitable metrics $d_\mathcal{I}$ and $d_\mathcal{S}$, and it is not always clear how to choose these metrics so that we can obtain meaningful results. Additionally, it is not always clear how to handle optimization problems with continuous input and discrete solutions: when using Equation~\ref{eq:continuousstability} to define stability, any change in solution would result in unbounded stability. However, we would still like to distinguish between solutions that change often from solutions that change rarely, which can be considered more stable.

As we have seen above, the trade-off between stability and solution quality is mainly influenced by the following three aspects: (1) how often the combinatorial structure of a solution changes, (2) whether there are discrete changes in a solution or only continuous changes, and (3) how far apart different solutions are, or how fast one solution can transform into another. In our framework we aim to overcome the drawbacks of the above described formulation of stability, by tackling the three aspects influencing the trade-off in isolation. Hence to analyze the trade-off between stability and solution quality, we propose to measure how the solution quality is affected by imposing various requirements on the stability of an algorithm. These requirements are formalized in the following three types of stability.

\begin{description}
	\item[Event stability] follows the setting of kinetic data structures. That is, the input (a~set of moving objects) changes continuously as a function over time. However, contrary to typical KDSs where a constraint is imposed on the solution quality, we aim to enforce the stability of the algorithm. For event stability we disallow the algorithm to change the solution too often. Doing so directly is problematic, but we formalize this approach using the concept of $k$-optimal solutions. We then obtain a trade-off between stability and quality that can be tuned by the parameter $k$. Note that event stability captures only \emph{how often} a solution changes, but not \emph{how much} a solution changes at each event.
	
	\item[Topological stability] takes a first step towards the definition of stability described in Equation~\ref{eq:continuousalgorithmstability}. However, instead of measuring the amount of change using a metric, we merely require the solution to behave continuously. To do so we need to define only a topology on the solution space $\mathcal{S}$ that captures stable behavior. We work under the assumption that the input follows a continuous path through the input space over time, and thus there is also a topology on the input space defining the stable behavior of the input. Surprisingly, even though we ignore the amount of change in a single time step, this type of analysis still provides meaningful information on the trade-off between solution quality and stability. In fact, the resulting trade-off can be seen as a lower bound for any analysis involving metrics that follow the used topology.
	
	\item[Lipschitz stability] fully captures the definition in Equation~\ref{eq:continuousalgorithmstability}. As the name suggests, it is inspired by Lipschitz continuity: if we see an algorithm as a function from input to output, then this function should be Lipschitz continuous. We provide an upper bound on the Lipschitz constant $K$, which restricts the speed at which the solution can change, relative to the change in the input. This further restricts the stability compared to topological stability, where there is no bound on the speed at which the solution could change. We are then again interested in the quality of the solutions that can be obtained with any $K$-Lipschitz stable algorithm. Given the complexity of this type of analysis, a complete trade-off for any value of the Lipschitz constant $K$ is typically out of reach, but results for sufficiently small or large values can still be of significant interest.
\end{description}

Lipschitz stability follows the algorithm stability definition introduced earlier and is thus the preferred type of analysis. However, its reliance on metrics $d_\mathcal{I}$ and $d_\mathcal{S}$ and speed parameter $K$ often makes Lipschitz stability analysis prohibitively challenging or infeasible. The other types of stability analysis simplify the stability requirement, making the analysis significantly easier. Specifically, topological stability analysis only relies on topologies on input and solution space, and for event stability analysis we need to count only the number of changes in a solution. Although these types of analysis do not fully capture all aspects of stability, they do offer useful insight into the interplay between problem instances, solutions, and the optimization function, which is invaluable for the development of stable algorithms.

\subsection{Algorithmic models}
\label{subsec:algmodels}
When applying either of the above definitions to analyze stability, we should always be aware of the algorithm at hand. Algorithms for time-varying input can adhere to different models, which are differentiated by the availability of the input. The model of an algorithm~$\mathcal{A}$ influences the results of the stability analysis. For input $I(t)$ depending on time $t$, we distinguish the following models.

\begin{description}
	\item[Stateless algorithms] depend only on the input $I(t)$ at a particular point in time, and no other information of earlier times. This in particular means that if $I(t_1) = I(t_2)$, then $\mathcal{A}$ produces the same output at time $t_1$ and at time $t_2$. These algorithms are essentially functions from input to output.
	
	\item[State-aware algorithms] have access not only to the input $I(t)$ at a particular time, but also maintain a state over time; in practice this is typically the output at the previous point in time. Thus, even if $I(t_1) = I(t_2)$, then $\mathcal{A}$ may produce different results at time $t_1$ and $t_2$ if the states at those times are different. 
	
	\item[Clairvoyant algorithms] have access to the complete function $I(t)$ and can adapt to future inputs. Thus, the complete output can be computed offline. 
\end{description}

\begin{figure}
	\centering
	\includegraphics{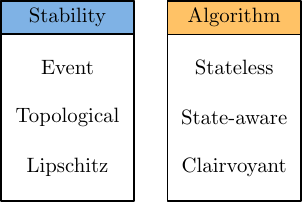}
	\caption{Overview of the stability types and algorithmic models defined by our framework.}
	\label{fig:framework-overview}
\end{figure}

While this distinction between types of algorithms resembles characterizations made in other areas of algorithmic research, we chose these names to signify how stability is influenced by the type of algorithm. For example, the state of an algorithm can be seen as knowing its \emph{history}, which is a term used in kinetic data structures to explain how the motion of the input up to a certain point in time influences the current state of the data structure. However, for stability we are interested only in recent history, as we want to ensure that the current solution does not change too quickly. Furthermore, clairvoyant algorithms are often called \emph{offline} algorithms in terms of streaming algorithms and dynamic data, since the complete output over time can be computed offline. However, for stability the clairvoyant aspect of this model is most important: adapting solutions to future inputs to minimize instabilities. In addition, stateless and state-aware algorithms are often called \emph{online} algorithms\footnote{Online algorithms generally consider memory efficiency as an additional quality criterion. Like running time, such considerations are beyond the scope of this work.} for dynamic data, or algorithms in the \emph{black-box} model for kinetic data structures, as the complete input over time is not known in advance. While this property indeed influences stability, we choose to make an additional distinction between stateless and state-aware: the stability between these models can differ, as shown in Sections~\ref{sec:eventstable} and~\ref{sec:topostable}.

Figure~\ref{fig:framework-overview} gives an overview of all parts of the framework.  
State-aware algorithms are arguably the most interesting of the three models for the following reasons. For many applications, for example online route planning, the input data is not available in full beforehand, rendering clairvoyant algorithms infeasible in such cases. Furthermore, a stateless algorithm cannot utilize the history to ensure stability, which sometimes results is worse stability compared to state-aware algorithms (see, for example, Section~\ref{sec:topostable}).
We conclude this section with a discussion on how to use our framework. 

\subsection{Applying the framework}
As described earlier in this section, the most intuitive way to analyze the trade-off between stability and solution quality is as follows. We enforce stability on the solutions to an optimization problem $\Pi$, and measure how high the solution quality can still be under this restriction. In our framework, this is done by choosing a type of stability along with a model for the algorithm solving the problem, and analyzing the solution quality that can be achieved. Preferably, we can still use optimal solutions, but in case those are unstable, we want to approximate an optimal solution as well as possible.

In general there are two approaches to find stable solutions with high quality: (1) we can start from an optimal solution $\OPT$ and use optimization function $f$ of $\Pi$ to look for solutions close to $\OPT$ that are stable, or (2) we can try to define a structure that is inherently stable according to our chosen type of stability, and analyze its solution quality using optimization function $f$. While we generally stick to the first approach, for stateless algorithms, this approach is not a viable strategy: guided by $f$, a stateless algorithm always finds $\OPT$, since it can access only the data at a particular point in time and use $f$ to look for an optimal solution. For stateless algorithms we therefore use the latter approach, which results in a continuous function when analyzing topological or Lipschitz stable solutions, as a stateless algorithm is a function and discontinuities are not allowed for these stability types.

The stable structures in the second approach are often inherently stateless. An example of this can be found in the stability analysis of the kinetic Euclidean 2-center problem by Durocher and Kirkpatrick~\cite{durocher2006steiner}. In the kinetic Euclidean 2-center problem, a set of moving points should be covered by two disks of minimum radius. They define reflection-based 2-centers, as a stable alternative to optimal Euclidean 2-centers. To find a reflection-based 2-center, we place a single disk and reflect its position through a reflection center, a central location in the point set. This can be computed on a static input, and hence a stateless algorithm can find the same solution on a time-varying input. To do so, the solution space $\mathcal{S}$ must be restricted to solution space $\mathcal{S}'$, consisting only of stable solutions, such as reflection-based 2-centers. A stateless algorithm can then find a stable optimal solution $\OPT'$ in $\mathcal{S}'$, and ideally $\OPT'$ should approximate the solution quality of $\OPT$. Note that in the other algorithmic models, this stateless computation can be mimicked, since those models are less restrictive in their access to data.

In practice, we often use a combination of the two approaches: the first approach, guided by the optimization function $f$ of $\Pi$, requires analysis of optimal and approximate solutions, and hence usually results in the most theoretical insights into the stability of those solutions. Even though the second approach is less straightforward and guided, the (often stateless) nature of this approach can simplify the stability analysis, making it a viable alternative when results are otherwise hard to achieve.

\begin{table}
    \centering
    \caption{Our results for combinations of stability types and algorithmic models.}
    \label{tab:results}
    \smallskip
    \begin{tabular}{c||c|>{\centering}m{3.5cm}|c|}
          & Event stability & Topological stability & Lipschitz stability 
         \\
         \hline
         \rule{0pt}{4ex} 
         Stateless algorithms & Section~\ref{sec:eventEMST-stateless} & Section~\ref{sec:topoEMST-stateless} & $\leftarrow$ Lower bound \\
         \rule{0pt}{4ex} 
         State-aware algorithms & Section~\ref{sec:eventEMST} & Section~\ref{sec:topoEMST} & Section~\ref{sec:lipschitzemst} \\
         \rule{0pt}{4ex} 
         Clairvoyant algorithms & Upper bounds $\uparrow$ & $\uparrow$ & $\uparrow$ \\[1.5ex]
         \hline
    \end{tabular}
\end{table}

Since both of the suggested approaches have their merits and drawbacks, we utilize both for our results in the upcoming sections. Specifically, we apply our framework to the kinetic Euclidean minimum spanning tree problem and analyze the solution quality achieved by algorithms adhering to different stability types and following various algorithmic models. We ensure that each combination of types and models is covered; see Table~\ref{tab:results} for an overview.

Note that certain results cover multiple cells in the table.  For example, topological stability can be seen as Lipschitz stability with an unbounded Lipschitz constant $K$. The solution quality obtained by a $K$-Lipschitz stable algorithm, with bounded~$K$, is therefore at least as bad as for a topologically stable algorithm. This means that lower bounds on solution quality for topological stability extend to the Lipschitz stability column. Similarly, clairvoyant algorithms can mimic state-aware algorithms and forgo looking into the future. Thus, upper bounds on solution quality obtained for state-aware algorithms extend to the clairvoyant algorithms row. In Section~\ref{sec:topostable-analysis}, we elaborate further on the relations between stability types and algorithmic models, and how these relations influence proving strategies.

\section{Event stability}\label{sec:eventstable}

The least restrictive form of stability is event stability. Like the number of external events in KDSs, event stability captures only how often the solution changes.

\subsection{Event stability analysis}
Let $\Pi$ be an optimization problem with a set of input instances $\mathcal{I}$, a set of solutions $\mathcal{S}$, and optimization function $f\colon \mathcal{I}\times\mathcal{S} \rightarrow \reals$. Following the framework of kinetic data structures, we assume that the input instances include certain parameters that change as a function of time, such as point coordinates. To apply event stability analysis, we require that all solutions have a combinatorial description, that is, the solution description does not use the time-varying parameters of the input instance. We further require that every solution $S \in \mathcal{S}$ is feasible for every input instance $I \in \mathcal{I}$. Insertions or deletions of elements can break this assumption: a spanning tree on $n$ points is not a feasible solution for $n+1$ point, since one edge is missing. Note that an insertion or a deletion would typically force an event, and a kinetic data structure would be allowed to recompute. Thus it makes sense to apply our event stability analysis only between such insertions and deletions.

For example, in the setting of kinetic EMSTs, the input instances would consist of a fixed set of points. The coordinates of these points then change as a function over time. A solution of the kinetic EMST problem consists of the combinatorial description of a tree on the set of input points. Note that every tree describes a feasible solution for any input instance, if we do not insist on any additional restrictions like, e.g., planarity. The minimization function $f$ then simply measures the total length of the tree, which does depend on the time-varying parameters of the problem instance.

We want to restrict how often a solution changes, in such a way that the solution is still close to an optimal solution in solution quality. Instead of doing so directly, we introduce the concept of $k$-optimal solutions. Let $d_\mathcal{I}$ be a metric on the input instances, and let $\OPT\colon \mathcal{I} \rightarrow \mathcal{S}$ describe the optimal solutions. We say that a solution $S \in \mathcal{S}$ is \emph{$k$-optimal} for an instance $I \in \mathcal{I}$ if there exists an input instance $I' \in \mathcal{I}$ such that $f(I', S) = f(I', \OPT(I'))$ and $d_\mathcal{I}(I, I') \leq k$. Any optimal solution is therefore always $0$-optimal.
We need to point out that the above definition requires a form of normalization on the metric $d_\mathcal{I}$, similar to that of, e.g., smoothed analysis~\cite{spielman2004smoothed}. We hence require that there exists a constant $c$ such that every solution $S \in \mathcal{S}$ is $c$-optimal for every instance $I \in \mathcal{I}$. For technical reasons we require the latter condition to hold only for some time interval $[0, T]$ of interest. Note that the concept of $k$-optimal solutions is closely related to \emph{backward error analysis} in numerical analysis \cite{wilkinson2003error}.% as well as $\gamma$-stable instances defined by Bilu and Linial~\cite{DBLP:journals/cpc/BiluL12}.

Following the framework of kinetic data structures, we typically require the functions of the time-varying parameters to be well-behaved (e.g., polynomial functions), for otherwise we cannot derive meaningful bounds. The event stability analysis then considers two aspects. First, we analyze how often the solution needs to change to maintain a $k$-optimal solution for every point in time. Second, we analyze how well a $k$-optimal solution approximates an optimal solution. Intuitively, $k$-optimal solutions approximate optimal solutions well, as the input did not change much since the solution was actually optimal. We can then enforce stability by ignoring all events that happen, required to keep the solution optimal, as long as a solution stays $k$-optimal. Once the solution is no longer $k$-optimal, we recompute to find the optimal solution, and repeat the process. This reduces the number of events, while the solution quality becomes only moderately worse. Typically we require certain reasonable (input) assumptions, to obtain good approximation bounds (as a function of $k$).

\subsection{Event stability of state-aware algorithms for EMSTs}\label{sec:eventEMST}
Our input consists of a set of points $P = \{p_1, \ldots, p_n\}$ where each point $p_i$ describes a trajectory by the function $p_i\colon [0, T] \rightarrow \reals^d$. The goal is to maintain a combinatorial description of a short spanning tree on $P$ that does not change often. We assume that the functions $x_i$ are polynomials with bounded degree $s$.

To use the concept of $k$-optimal solutions, we first need to normalize the coordinates. We simply assume that $p_i(t) \in [0, 1]^d$ for $t \in [0, T]$. This assumption may seem overly restrictive for kinetic point sets, but note that we are interested only in relative positions, and thus the frame of reference may move with the points. Next, we define the metric $d_\mathcal{I}$ along the trajectory of all points as follows.
\begin{equation}\label{eq:Euclmetric}
d_\mathcal{I}(I, I') = \max_i \|p_i - p'_i\|
\end{equation}
Note that this metric, and the resulting definition of $k$-optimal solutions, is not specific to EMSTs and can be used in general for problems with kinetic point sets as input. In our case $\|a-b\|$ denotes the distance between $a$ and $b$ in the (Euclidean) $\ell_2$ norm. Now let $S(t) = \OPT(t)$ be the EMST at time $t$. Then, by definition, $S(t)$ is $k$-optimal at time $t'$ if $d_\mathcal{I}(I(t), I(t')) \leq k$. As explained before, our approach is now very simple: we compute the EMST and keep that solution as long as it is $k$-optimal, after which we compute the new EMST, and so forth. Below we analyze this approach.
%how often we need to recompute the EMST, and how well a $k$-optimal solution approximates the EMST.

\paragraph{Number of events.}
To bound the number of events, we first need to bound the speed of any point with a polynomial trajectory and bounded coordinates. For this we can use a classic result known as the \emph{Markov Brothers' inequality}.
\begin{lemma}[\cite{markoff1916polynome}]\label{thm:markovbrothers}
	Let $h(t)$ be a polynomial with degree at most $s$ such that $h(t) \in [0, 1]$ for $t \in [0, T]$, then $|\mathrm{d}h(t)/\mathrm{d}t|\leq s^2 / T$ for all $t \in [0, T]$.
\end{lemma}
\begin{lemma}\label{lem:eventupper}
	For a kinetic point set $P$ with degree-$s$ polynomial trajectories $p_i(t) \in [0, 1]^d$ ($t \in [0, T]$) we need $O(\frac{s^2}{k})$ changes to maintain a $k$-optimal solution for constant $d$.
\end{lemma}
\begin{proof}
	By Lemma~\ref{thm:markovbrothers} the velocity of any point is at most $s^2 / T$ in one dimension, and thus at most $\sqrt{d} \, s^2 / T = O(s^2/T)$ in $d$ dimensions, assuming $d$ is constant. Now assume that we have computed an optimal solution $S$ for some time $t$. The solution $S$ remains $k$-optimal until one of the points has moved at least $k$ units. Since the velocity of the points is bounded, this takes at least $\Delta t = \Omega(k T / s^2)$ time, at which point we can recompute the optimal solution. Since the total time interval is of length $T$, this can happen at most $T / \Delta t = O(s^2 / k)$ times.
\end{proof}
Next we show that this upper bound is tight up to a factor of $s$, using \emph{Chebyshev polynomials} of the first kind \cite{rivlin1974chebyshev}. A Chebyshev polynomial of degree $s$ with range $[0, 1]$ and domain $[0, T]$ passes through the entire range exactly $s$ times.
\begin{figure}[b]
	\centering
	\includegraphics{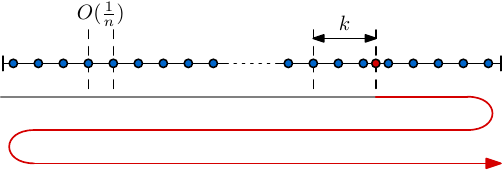}
	\caption{The blue points are stationary, while the red point moves along a trajectory of degree $3$, triggering $\Omega(\frac{s}{k})$ events. The trajectory is shown as an arrow along the 1D space; the gray part has already been traversed.}
	\label{fig:polyLB}
\end{figure}
\begin{lemma}\label{lem:eventlower}
	For a kinetic point set $P$ of $n$ points with degree-$s$ polynomial trajectories $p_i(t) \in [0, 1]^d$ ($t \in [0, T]$) we need $\Omega(\min(\frac{s}{k}, s n))$ changes in the worst case to maintain a $k$-optimal solution.
\end{lemma}
\begin{proof}
	We can restrict ourselves to $d=1$. Let $p_1$ move along a Chebyshev polynomial of degree $s$, and let the remaining points be stationary and placed equidistantly along the interval $[0, 1]$. As soon as $p_1$ meets one of the other points, then $p_1$ can travel at most $k$ units before the solution is no longer $k$-optimal (see Figure~\ref{fig:polyLB}). Therefore, $p_1$ moving through the entire interval requires $\Omega(\min(1/k, n))$ changes to the solution. Doing so $s$ times gives the desired bound.
\end{proof}

For points following polynomial trajectories, we hence proved the following theorem.

\begin{theorem}\label{thm:polynomial-trajectory-events}
    For a kinetic point set $P$ of $n$ points with degree-$s$ polynomial trajectories $p_i(t) \in [0, 1]^d$ ($t \in [0, T]$) we need between $O(\frac{s^2}{k})$ and $\Omega(\min(\frac{s}{k}, s n))$ changes in the worst case to maintain a state-aware $k$-optimal solution.
\end{theorem}
It is important to notice that this lower bound behavior is fairly specific to polynomial trajectories. If we allow more general trajectories, then we can prove a stronger lower bound.
\begin{lemma}
	For a kinetic point set $P$ of $n$ points with degree-$s$ pseudo-algebraic trajectories $p_i(t) \in [0, 1]^d$ ($t \in [0, T]$) we need $\Omega(\min(\frac{sn}{k}, sn^2))$ changes in the worst case to maintain a $k$-optimal solution.
\end{lemma}
\begin{proof}
	We can restrict ourselves to $d=1$. Any two pseudo-algebraic trajectories of degree at most $s$ can cross each other at most $s$ times. We make $n/2$ points stationary and place them equidistantly along the interval $[0, 1]$. The other $n/2$ points follow trajectories that take them through the entire interval $s$ times, in such a way that every point moves through the entire interval completely before another point does so. The resulting trajectories are clearly pseudo-algebraic, and each time a point moves through the entire interval it requires $\Omega(\min(1/k, n))$ changes to the solution. As a result, the total number of changes is $\Omega(\min(\frac{sn}{k}, sn^2))$.
\end{proof}
We can show the same lower bound for algebraic trajectories of degree at most $s$, but this is slightly more involved.
\begin{figure}[b]
	\centering
	\includegraphics{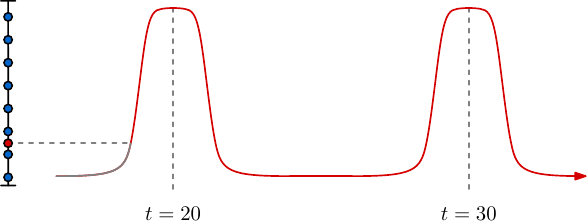}
	\caption{The blue points are stationary, while the red point moves along a trajectory of degree $8$ and is the second point to start moving through the blue points ($p_1(t) = \sum_{j=0}^{2} \frac{1}{(t-10\cdot j-20)^4+1}$). The trajectory is shown as an arrow along the 1D space; the gray part has already been traversed.}
	\label{fig:algebraicLB}
\end{figure}
\begin{lemma}
	For a kinetic point set $P$ of $n$ points with degree-$s$ algebraic trajectories $p_i(t) \in [0, 1]^d$ ($t \in [0, T]$) we need $\Omega(\min(\frac{sn}{k}, sn^2))$ changes in the worst case to maintain a $k$-optimal solution.
\end{lemma}
\begin{proof}
	We can restrict ourselves to $d=1$. We make $n/2$ points stationary and place them equidistantly along the interval $[0, 1]$. The other $n/2$ points follow trajectories that take them through the entire interval $s/4$ times, in such a way that every point moves through the entire interval completely before another point does so. The trajectory of a non-stationary point is $p_i(t) = \sum_{j=0}^{s/4} \frac{1}{(t-10\cdot j-10\cdot i\cdot s/4)^4+1}$. The trajectory consists of $s/4$ moves through the stationary points, one such move every $10$ time units (see Figure~\ref{fig:algebraicLB}). The $i$-th point will be finished $10* s/4$ time units after it starts its first move through the stationary points, while the $(i+1)$-st point starts $10$ units after the $i$-th point finishes. The resulting trajectories are clearly algebraic, and each time a point moves through the entire interval it requires $\Omega(\min(1/k, n))$ changes to the solution. As a result, the total number of changes is $\Omega(\min(\frac{sn}{k}, sn^2))$.
\end{proof}
\paragraph{Solution quality.}  To analyze the solution quality of $k$-optimal solutions, we prove a bound on the ratio between the length of $k$-optimal solutions and the length of optimal EMSTs. In general, we cannot expect $k$-optimal solutions to be a good approximation of an optimal EMST's length: if all points are within distance $k$ from each other, then all solutions are $k$-optimal. We therefore need to make the assumption that the points are spread out reasonably throughout the motion. To quantify this, we use a measure inspired by the \emph{order-$l$ spread}, as defined by Erickson \cite{erickson2005dense}. Let $\prop{mindist}_l(P)$ be the smallest distance in $P$ between a point and its $l$-th nearest neighbor. We assume that $\prop{mindist}_l(P) \geq 1/\Delta_l$ throughout the motion, for some value of $\Delta_l$. We can use this assumption to give a lower bound on the length of the EMST. Pick an arbitrary point and remove all points from $P$ that are within distance $1/\Delta_l$, and repeat this process until the smallest distance is at least $1/\Delta_l$. By our assumption, we remove at most $l-1$ points for each chosen point, so at least $n/l$ points are left. The distance between each pair of points is now at least $1/\Delta_l$, hence an EMST on the remaining points has length $\Omega(\frac{n}{l\Delta_l})$. This readily forms a lower bound on the length of the EMST on $P$: even if all removed points were Steiner points, adding back those points can improve the length of an EMST only by a constant factor \cite{chung1985new}.

\begin{lemma}\label{lem:eventapprox}
	A $k$-optimal solution of the EMST problem on a set of $n$ points $P$ is an $O(1 + kl\Delta_l)$-approximation of the EMST, under the assumption that $\prop{mindist}_l(P) \geq 1/\Delta_l$.
\end{lemma}
\begin{proof}\belowdisplayskip=-12pt
	Let $S$ be a $k$-optimal solution of $P$ and let $\OPT$ be an optimal solution of $P$. By definition there is a point set $P'$, with $d_\mathcal{I}(P, P') \leq k$, for which the length of solution $S$ is at most that of the optimal solution $\OPT'$ of $P'$. Since $d_\mathcal{I}(P, P') \leq k$, the length of each edge can grow or shrink by at most $2k$ when moving from $P'$ to $P$. Therefore we can state that $f(P, S) \leq f(P, \OPT) + 4 k n$, as in the worst case every edge in the $k$-optimal solution can have grown by $2k$ while in an optimal solution every edge can have shrunk by $2k$. Now, using the lower bound of $\Omega(\frac{n}{l\Delta_l})$ on the length of an EMST, we obtain the following.
	\begin{align*}
	f(P, \OPT) + 4kn &\leq f(P, \OPT) + 4k O(f(P, \OPT) l \Delta_l) \\
	&= O(1 + k l \Delta_l)\cdot f(P, \OPT)
	\end{align*}
\end{proof}
Note that there is a clear trade-off between the approximation ratio and how restrictive the assumption on the spread is. Regardless, we can obtain a decent approximation, while processing only a small number of events. Choosing reasonable values $k = O(1/n)$, $l = O(1)$, and $\Delta_l = O(n)$, then, under the assumptions, a constant-factor approximation of an EMST can be maintained while processing only a linear number of events, extending Theorem~\ref{thm:polynomial-trajectory-events} to:

\begin{theorem}
    For a kinetic point set $P$ of $n$ points with constant-degree polynomial trajectories $p_i(t) \in [0, 1]^d$ ($t \in [0, T]$) we need $\Theta(n)$ changes in the worst case to maintain a state-aware $O(1)$-approximation of the EMST, under the assumption that $\prop{mindist}_{O(1)}(P) \geq O(1/n)$.
\end{theorem}

\subsection{Event stability of stateless algorithms for EMSTs}\label{sec:eventEMST-stateless}
In the previous section, we showed how a constant-factor approximation of an EMST can be achieved with state-aware algorithms that process at most a linear number of events. In this section, we extend those results to stateless algorithms, by fitting a grid to the input space and rounding all points in $P$ to this grid before computing an EMST.

\paragraph{Grid construction.} We place an axis-aligned grid $G$ over the input points~$P$, where each grid cell is a box with side length $k$. Before computing an EMST, we round the coordinates of each point $p\in P$ to the nearest grid point $g$ as follows. In two dimensions, a point $p = (x, y)$ is rounded to $g = (a, b)$ if $-\frac{k}{2} \leq x - a < \frac{k}{2}$ and $-\frac{k}{2} \leq y - b < \frac{k}{2}$. Note that the region in which $p$ will be rounded to $g$ also forms a box with sides of length $k$, which is centered around $g$, closed at the bottom and left side, and open at the top and right side. These boxes around the grid cells partition the plane into regions, in which points round to the same grid point.

\paragraph{Number of events.} We first show that the bounds on the number of events in the previous section extend. For the previous bounds we used $k$-optimal solutions instead of rounding to a grid, and hence we now show that this leads to only small differences.

\begin{lemma}\label{lem:eventupper-stateless}
	For a kinetic point set $P$ of $n$ points with degree-$s$ polynomial trajectories $p_i(t) \in [0, 1]^d$ ($t \in [0, T]$) we need $O(\frac{s^2n}{k} + sn)$ changes to maintain a $k$-optimal solution for constant $d$.
\end{lemma}
\begin{proof}
	We again use Lemma~\ref{thm:markovbrothers} to find that the velocity of any point is at most $s^2 / T$ in one dimension. Now assume that we have computed a solution $S$ for point set $P'$, the original points $P$ rounded to grid points in $G$. When a point in $P$ is rounded to a different grid point, $P'$ changes and we have to recompute. First consider only a single dimension. A point $p\in P$ has to move $k$ units in this dimension to completely traverse the region that rounds to one grid point. Traversing such a region hence takes at least $\Delta t = \Omega(k T / s^2)$ time, and in a time interval of length $T$, at most $O(s^2/k)$ events take place.
	
	However, starting from the initial position of $p$, it may take less than $k$ units to swap the rounding to a different grid point. Similarly, since the points follow degree-$s$ polynomials, every grid line is crossed at most $s$ times. Thus, a point $p$ can cross back and forth over a grid line, without crossing a different grid line in between, at most $s$ times. In these instances, $p$ can swap the point it is rounding to without moving $k$ units. In total, a point $p$ can trigger $O(s^2/k + s)$ events. Thus, for constant dimension $d$, there are $O(\frac{s^2n}{k} + sn)$ events.
\end{proof}
Just as in the state-aware case, this upper bound is tight up to a factor of $s$.
\begin{lemma}\label{lem:eventlower-stateless}
	For a kinetic point set $P$ of $n$ points with degree-$s$ polynomial trajectories $p_i(t) \in [0, 1]^d$ ($t \in [0, T]$) we need $\Omega(\frac{sn}{k} + sn)$ changes in the worst case to maintain a $k$-optimal solution.
\end{lemma}
\begin{proof}
	We can restrict ourselves to $d=1$. Let all points move along Chebyshev polynomials of degree $s$. As soon as a point $p\in P$ changes which grid point it is rounded to, then $p$ can travel at most $k$ units before again changing its rounding. Therefore, $p$ moving through the entire grid once requires $\Omega(1/k)$ changes to the solution. Additionally, $p$ changes direction at most $s$ times, and every time this happens, another change in rounding is triggered. Thus moving through the grid $s$ times, along with the $s$ additional events for changing direction, gives the desired bound for each of the $n$ points.
\end{proof}

\paragraph{Solution quality.} The bound on the solution quality in Lemma~\ref{lem:eventapprox} extends to our new setting, under the same assumption that $\prop{mindist}_l(P) \geq 1/\Delta_l$.  To see this, let $S$ be an optimal solution for $P'$, the input points $P$ rounded to grid points in $G$. By definition of $G$, the points in $P'$ are at most $\frac{\sqrt{d}}{2}k$ removed from their original position in $P$, thus $d_\mathcal{I}(P, P') \leq \frac{\sqrt{d}}{2}k$. Hence, solution $S$ is $\frac{\sqrt{d}}{2}k$-optimal for $P$, and Lemma~\ref{lem:eventapprox} directly applies.

We can again choose constants $k = O(1/n)$, $l = O(1)$, and $\Delta_l = O(n)$, to get the following theorem for stateless EMST algorithms in constant dimensions.

\begin{theorem}
    For a kinetic point set $P$ of $n$ points with constant-degree polynomial trajectories $p_i(t) \in [0, 1]^d$ ($t \in [0, T]$) we need $\Theta(n^2)$ changes in the worst case to maintain a stateless $O(1)$-approximation of the EMST, under the assumption that $\prop{mindist}_{O(1)}(P) \geq O(1/n)$.
\end{theorem}

\section{Topological stability}\label{sec:topostable}
The event stability analysis has two major drawbacks: (1) it is applicable only to problems for which the solutions are always feasible and described combinatorially, and (2) it does not distinguish between small and large structural changes. Topological stability analysis is applicable to a wide variety of problems and enforces continuous changes to the solution.

\subsection{Topological stability analysis}\label{sec:topostable-analysis}
Let $\Pi$ be an optimization problem with input instances $\mathcal{I}$, solutions $\mathcal{S}$, and optimization function $f$. An algorithm $\mathcal{A}$ is \emph{topologically stable} if, for any continuous path $I\colon [0,1] \rightarrow \mathcal{I}$ in $\mathcal{I}$, $\mathcal{A}$ maps it to a continuous path $S$ in $\mathcal{S}$. To properly define a continuous path in $\mathcal{I}$ and $\mathcal{S}$, we need to specify a topology $\mathcal{T}_\mathcal{I}$ on $\mathcal{I}$ and a topology $\mathcal{T}_\mathcal{S}$ on $\mathcal{S}$. An overview of this model can be found in Figure~\ref{fig:topological-input-solution-space}. Alternatively, we could specify metrics $d_\mathcal{I}$ and $d_\mathcal{S}$, but this is typically more involved. Let $\mathcal{P}_\mathcal{I}$ be the set of continuous paths through $\mathcal{I}$. We then want to analyze the approximation ratio $\TS$ of any topologically stable algorithm with respect to $\OPT$, which we will call the \emph{topological stability ratio}. That is, we are interested in the ratio
\begin{figure}
	\centering
	\includegraphics{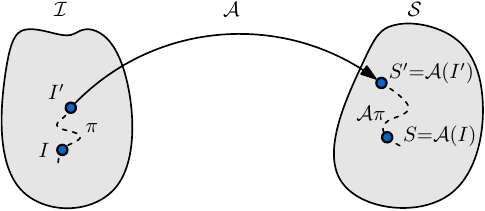}
	\caption{Algorithm $\mathcal{A}$ maps input instances from the input space $\mathcal{I}$ to the solution space $\mathcal{S}$. Continuous path $I$ in the space defined by topology $\mathcal{T}_\mathcal{I}$, is mapped to continuous path $S$ in the space defined by topology $\mathcal{T}_\mathcal{S}$.}
	\label{fig:topological-input-solution-space}
\end{figure}
\begin{equation}\label{eq:topo}
\TS(\Pi, \mathcal{T}_\mathcal{I}, \mathcal{T}_\mathcal{S}) = \inf_{\mathcal{A}} \sup_{I \in \mathcal{P}_\mathcal{I}} \sup_{t \in [0,1]} \frac{f(I(t), S(t))}{f(I(t), \OPT(I(t)))}
\end{equation}
where the infimum is taken over all topologically stable algorithms. Naturally, if $\OPT$ is already topologically stable, then this type of analysis does not provide any insight and the ratio is simply $1$. However, $\OPT$ is not topologically stable if it undergoes discrete changes, and in that case topological stability allows you to measure what solution quality can be achieved by requiring continuity.

The above analysis can even be applied when the solution space (or the input space) is discrete. In such cases, continuity can often be defined using the graph topology of so-called flip graphs, for example, based on edge flips for triangulations or rotations in rooted binary trees. The vertices of such a graph each represent a solution (or input) with a different combinatorial structure, while the edges represent the possible transitions between solutions. To create a continuous solution space, we still represent the discrete space using the vertices of the flip graph, but we create continuity on the edges: we define a (continuous) topological space by representing vertices by points, and representing every edge of the graph by a copy of the unit interval $[0,1]$. These intervals are glued together at the vertices. In other words, we consider the corresponding simplicial $1$-complex.

For EMSTs we can do exactly what we just described, since the solution space is discrete and the vertices of a flip graph represent spanning trees. Although the points in the interior of the edges of this topological space do not represent proper spanning trees, we can still use this topological space in Equation~\ref{eq:topo} by extending $f$ over the edges via linear interpolation. This ensures that the value of $f$ for one of the vertices incident to an edge is as least as high as the value of $f$ anywhere on the edge. We therefore need to consider only the vertices of the flip graph (which represent proper spanning trees) to compute the topological stability ratio. Figure~\ref{fig:spanning-tree-space} shows an example of such a topological space for spanning trees.

\begin{figure}
	\centering
	\includegraphics{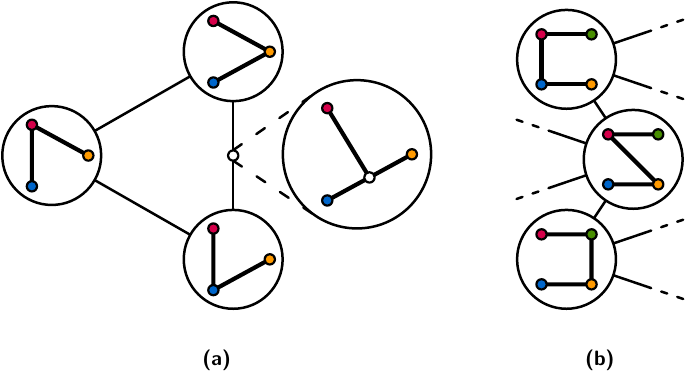}
	\caption{Topological spaces defined by flip graphs for edge slides/edge rotations on EMSTs. \textbf{\textsf{(a)}} The complete solution space for EMSTs on three vertices, along with an intermediate solution. \textbf{\textsf{(b)}} A partially drawn solution space for EMSTs on four vertices.}
	\label{fig:spanning-tree-space}
\end{figure}

\paragraph{Proving bounds on $\TS$.} First consider only state-aware algorithms. To prove an upper bound on the topological stability ratio, we have to describe a state-aware algorithm that produces a topologically stable solution. If this algorithm computes an $r$-approximation of the optimal solution, then we have found an upper bound of $r$ on $\TS$. While such an algorithm works on time-varying data, we usually consider static inputs that allow multiple optimal solutions, when proving an upper bound. The optimal solution would undergo a discrete change when such an input was encountered during motion. We define a continuous transformation that works for any such static input, and transforms one optimal solution of the static input into another. Note that the continuous transformation should follow the chosen topology $\mathcal{T}_\mathcal{S}$. If during this transformation, the solution is at most a factor $r$ worse than $\OPT$ according to $f$, then we immediately obtain an upper bound of $r$ on $\TS$: an algorithm can wait until a discrete change happens and then apply the transformation to produce an $r$-approximation. This approximation is topologically stable, since the transformation is continuous, and topological stability does not bound the speed at which the solution can change. Hence, at the point in time where the discrete change would happen, the algorithm may ``freeze time'' and apply the transformation to swap between optimal solutions.

For a lower bound on the topological stability ratio, we should consider a full time-varying input, for which an approximation ratio of $r$ is always necessary. However, we can again use static inputs that allow multiple optimal solutions, to simplify the analysis. We first construct a static input $I$, where every continuous transformation from one optimal solution to another, requires a solution that is at least a factor $r$ worse than $\OPT$. Thus, any algorithm computing a continuous transformation for $I$ produces at least an $r$-approximation. Finally, we find a motion of the input points in which $I$ occurs at some time $t$, and this motion should force the continuous transformation to happen somewhere during the motion. This is achieved by ensuring that keeping the same solution during the complete motion, or attempting a continuous transformation before/after $t$ results in a solution that is even worse. The best any algorithm can do, is transforming exactly at $t$, but this requires a solution that is at least a factor $r$ worse than $\OPT$. Thus every topologically stable algorithm computes at least an $r$-approximation on this time-varying input: $\TS$ is lower bounded by $r$.

Note that since we constructed a state-aware algorithm to prove an upper bound on the topological stability ratio $\TS$, we also prove a bound on the topological stability ratio of clairvoyant algorithms. A clairvoyant algorithm can simply emulate a state-aware algorithm, by not using future time steps. On the other hand, the method we describe to prove a lower bound on $\TS$ for state-aware algorithms, shows that \emph{every} topologically stable algorithm requires an $r$-approximation. Thus we prove a stronger statement than required: even a clairvoyant algorithm cannot do better.

For topologically stable stateless algorithms, which essentially are continuous functions between input space $\mathcal{I}$ and solution space $\mathcal{S}$, we take a different approach. For stateless algorithms we can analyze (an upper bound on) the approximation ratio of such a continuous function, for example by finding the function. However, usually we are interested in lower bounds, which motivate the usage of state-aware algorithms. A lower bound can be constructed by considering a continuous part of the input space and the possible mappings between those input instances and their solutions. We then prove a lower bound of $r$ on $\TS$ by analyzing the function that produces solutions according to this mapping. In particular, we prove a statement such as the following: if such a function maps a continuous (part of) the input to a continuous (part of) the solution space, then the produced solution must be at least an $r$-approximation of a respective optimal solution, or equivalently, if such a function produces a better solution than an $r$-approximation, then the mapping cannot be continuous.

\subsection{Topological stability of state-aware algorithms for EMSTs}\label{sec:topoEMST}
We use the same setting of the kinetic EMST problem as in Section~\ref{sec:eventEMST}, except that we do not restrict the trajectories of the points and we do not normalize the coordinates. We merely require that the trajectories are continuous. To define this properly, we need to define a topology on the input space, but for a kinetic point set with $n$ points in $d$ dimensions we can simply use the standard topology on $\reals^{dn}$ as $\mathcal{T}_\mathcal{I}$. To apply topological stability analysis, we also need to specify a topology on the (discrete) solution space. As the points move, the minimum spanning tree may have to change at some point in time by removing one edge and inserting another edge, which is called an \emph{edge flip}. From a practical point of view, we do not consider this operation to be stable, since the updated edge can be reinserted anywhere in the tree. From a theoretical point of view, the operation is also trivially stable; see below. Instead we define the topology of $\mathcal{S}$ using a flip graph, where the operations are either \emph{edge slides} or \emph{edge rotations}~\cite{aichholzer2002sequences,goddard1996distances,DBLP:conf/latin/NicholsPTZ18} (see Figure~\ref{fig:spanning-tree-space}). The optimization function $f$, measuring the quality of the EMST, is naturally defined for the vertices of the flip graph as the length of the spanning tree, and we use linear interpolation to define $f$ on the edges of the flip graph. For all three operations, we provide upper and lower bounds on $\TS(\text{EMST}, \mathcal{T}_\mathcal{I}, \mathcal{T}_\mathcal{S})$.

\paragraph{Edge flips.} One of the most general operations to update a spanning tree is the edge flip: one edge is replaced by another edge, that has one or both of its endpoints at different input points. After the operation, all input points must still be connected, to ensure the result is a spanning tree. For most applications that require stability, this operation is not considered stable, since the removed edge and the inserted edge may be far apart. When we use this operation to define the flip graph for the topology~$\mathcal{T_S}$ of~$\mathcal{S}$, and essentially deem this unstable operation to be stable, then the topological stability ratio is (predictably) very low.

\begin{theorem}\label{thm:state-aware-flips}
    For a state-aware algorithm~$\mathcal{A}$ solving the kinetic EMST problem, if $\mathcal{T}_\mathcal{S}$ is defined by edge flips, then $ \TS(\text{EMST}, \mathcal{T}_\mathcal{I}, \mathcal{T}_\mathcal{S}) =1$.
\end{theorem}
\begin{proof}
    Consider a point in time where the EMST has to be updated by removing an edge~$e$ and inserting an edge~$e'$, where $|e|=|e'|$. Observe that this operation is an edge flip and hence there is an edge in the flip graph defining $\mathcal{T}_\mathcal{S}$ connecting the (optimal) EMSTs before and after this operation. Since the optimization function is interpolated over this edge of the flip graph, and both end points have the same (optimal) function value, the optimal solution is also always a topologically stable solution.
\end{proof}

\paragraph{Edge slides.}
An edge slide is defined as the operation of moving one endpoint of an edge to one of its neighboring vertices along the edge to that neighbor. More formally, an edge $(u,v)$ in the tree can be replaced by $(u,w)$ if $w$ is a neighbor of $v$ and $w \neq u$. Note that a tree stays connected after edge slides. Since this operation is very local, we consider it to be stable. However, it more restrictive than edge flips, hence $\TS(\text{EMST}, \mathcal{T}_\mathcal{I}, \mathcal{T}_\mathcal{S})$ will be higher.

\begin{figure}
	\centering
	\includegraphics{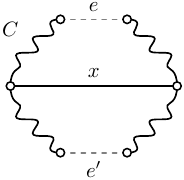}
	\caption{A configuration where $x$ is the longest edge when sliding from $e$ to $e'$.}
	\label{fig:TopStabSlideUB}
\end{figure}

\begin{lemma}\label{lem:slideupper}
	If $\mathcal{T}_\mathcal{S}$ is defined by edge slides, then \mbox{$\TS(\text{EMST}, \mathcal{T}_\mathcal{I}, \mathcal{T}_\mathcal{S}) \leq \frac{3}{2}$}.	
\end{lemma}
\begin{proof}
Consider a time where the EMST has to be updated by removing an edge $e$ and inserting an edge $e'$, where $|e| = |e'|$. Note that $e$ and $e'$ form a cycle $C$ with other edges of the EMST. We now slide edge $e$ to edge $e'$ by sliding its endpoints along the edges of $C$. Let $x$ be the longest intermediate edge when sliding from $e$ to $e'$ (see Figure~\ref{fig:TopStabSlideUB}). To allow $x$ to be as long as possible with respect to the length of the EMST, the EMST should be fully contained in $C$. By the triangle inequality we get that $2 |x| \leq |C|$. Since the length of the EMST is $\OPT = |C| - |e|$, we get that $|x| \leq \OPT/2 + |e|/2$. Thus, the length of the intermediate tree is $|C| - 2|e| + |x| = \OPT - |e| + |x| \leq \frac{3}{2}\OPT$.
\end{proof}

\begin{lemma}\label{lem:slidelower}
	If $\mathcal{T}_\mathcal{S}$ is defined by edge slides, then $\TS(\text{EMST}, \mathcal{T}_\mathcal{I}, \mathcal{T}_\mathcal{S}) \geq \frac{\pi+1}{\pi} \approx 1.318$.
\end{lemma}
\begin{proof}
	Consider a point in time where the EMST has to be updated by removing an edge $e$ and inserting an edge $e'$, where $|e|=|e'|$ is very small. Let the remaining points be arranged in a circle with diameter $d$, as shown in Figure~\ref{fig:slidelower}a. Furthermore, let $\OPT$ be the length of the EMST, then we get that $\OPT < d\pi$, since $\OPT$ cannot form a cycle around the circle. Simply using edge slides to move $e$ toward $e'$ will always grow $e$ to be length at least $d - \varepsilon$. We can make this construction for any $\varepsilon > 0$ by using sufficiently many points around the circle and consequently making $e$ and $e'$ arbitrarily short. Alternatively, $e$ can take a shortcut by sliding over another edge $e^*$ as a chord (see Figure~\ref{fig:slidelower}b). Doing so would require $e^*$ to first slide into this position. The shortcut is only beneficial if $|e|+|e^*|< d - \varepsilon$. However, if $e^*$ helps $e$ to avoid becoming a diameter of the circle, then $e$ and $e^*$, as chords, must span an angle larger than $\pi$ together. As a result, for a circle of diameter $d$, we get $|e|+|e^*|\geq d > d - \varepsilon$ by triangle inequality.
	
	\begin{figure}
		\centering
		\includegraphics{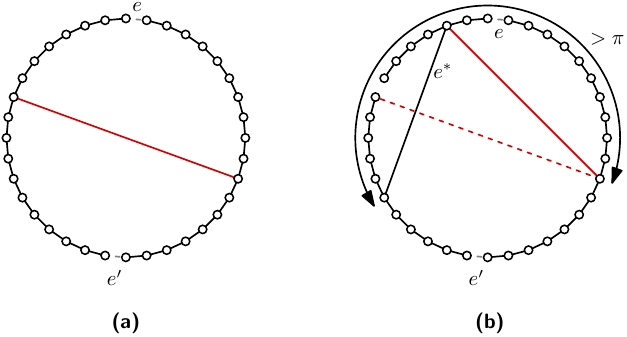}
		\caption{A $(\frac{\pi +1}{\pi} - \varepsilon)$-approximation of the EMST. \textbf{\textsf{(a)}} Edge $e$ slides to $e'$ and becomes the diameter of the circular configuration. \textbf{\textsf{(b)}}  Sliding edge $e^*$ to form a chord creates an even longer spanning tree.}
		\label{fig:slidelower}
	\end{figure}
	
	A motion of the points that forces $e$ to slide to $e'$ in this particular configuration looks as follows. The points start at $e$ and move at constant speed along the circle, half of the points clockwise and the other half counter clockwise. The speeds are assigned in such a way that at some point all points are evenly spread along the circle. Once all points are evenly spread, they start moving towards $e'$, again along the circle. Any edge sliding from $e$ to $e'$ during the motion must have length $d - \varepsilon$ at some point throughout the motion. On the other hand, $\OPT$ is largest when the points are evenly spread along the circle. Let the circle have diameter $d$, then $\OPT$ has length at most $d\pi - |e|$, and equivalently $d \geq \OPT/\pi + |e|$. Since we argued that the sliding edge will always have length $d - \varepsilon$ at some point, the largest intermediate solution has length at most $\OPT - |e| + d - \varepsilon$. Thus, for any small constant $\varepsilon > 0$, we show that $\TS(\text{EMST}, \mathcal{T}_\mathcal{I}, \mathcal{T}_\mathcal{S}) \geq \frac{\OPT - |e| + d - \varepsilon}{\OPT} \geq \frac{\OPT + \OPT/\pi - \varepsilon}{\OPT} \geq \frac{\pi +1}{\pi} - \varepsilon \approx 1.318 - \varepsilon$.
\end{proof}

\begin{theorem}
    For a state-aware algorithm~$\mathcal{A}$ solving the kinetic EMST problem, if $\mathcal{T}_\mathcal{S}$ is defined by edge slides, then $ 1.318\approx\frac{\pi+1}{\pi}\leq \TS(\text{EMST}, \mathcal{T}_\mathcal{I}, \mathcal{T}_\mathcal{S}) \leq \frac{3}{2}$.
\end{theorem}

\paragraph{Edge rotations.}
Edge rotations are a generalization of edge slides, that allow one endpoint of an edge to move to any other vertex. These operations are clearly not as stable as edge slides, but they are still more stable than the deletion and insertion of arbitrary edges. This is also reflected in our bounds on the topological stability ratio for edge rotations.

	\begin{figure}
		\centering
		\includegraphics{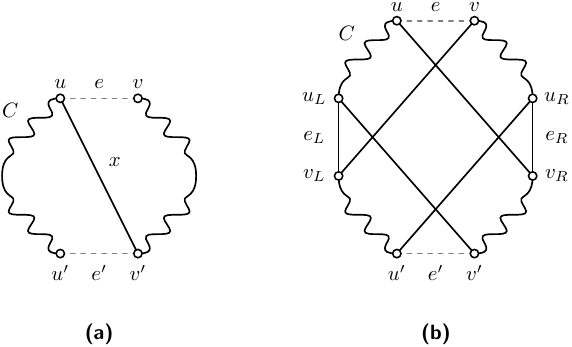}
		\caption{Potential intermediate edges when rotating $e$ to $e'$. \textbf{\textsf{(a)}} If one part of $C$ is small enough, then we can rotate one endpoint of $e$ directly to one endpoint of $e'$. \textbf{\textsf{(b)}} Otherwise, there are multiple options when using two edge rotations.}
		\label{fig:TopStabRotateUB}
	\end{figure}

\begin{lemma}\label{lem:nrotateupper}
	If $\mathcal{T}_\mathcal{S}$ is defined by edge rotations, then \mbox{$\TS(\text{EMST}, \mathcal{T}_\mathcal{I}, \mathcal{T}_\mathcal{S}) \leq \frac{4}{3}$}.	
\end{lemma}
\begin{proof}
	Consider a time where the EMST has to be updated by removing an edge $e = (u, v)$ and inserting an edge $e' = (u', v')$, where $|e| = |e'|$. Note that $e$ and $e'$ form a cycle $C$ with other edges of the EMST. We now rotate edge $e$ to edge $e'$ along some of the vertices of $C$. Let $x$ be the longest intermediate edge when optimally rotating from $e$ to $e'$. To allow $x$ to be as long as possible with respect to the length of the EMST, the EMST should be fully contained in $C$. We argue that $|x| \leq \OPT/3 + |e|$, where $\OPT$ is the length of the EMST. Removing $e$ and $e'$ from $C$ splits $C$ into two parts, where we assume that $u$ and $u'$ ($v$ and $v'$) are in the left (right) part. First assume without loss of generality that the left part has length at most $\OPT/3$. Then we can rotate $e$ to $(u, v')$, and then to $e'$, which implies that $|x| = |(u, v')| \leq \OPT/3 + |e|$ by the triangle inequality (see Figure~\ref{fig:TopStabRotateUB}a).
	
	Now assume that both parts have length at least $\OPT/3$. Let $e_L = (u_L, v_L)$ be the edge in the left part that contains the midpoint of that part, and let $e_R = (u_R, v_R)$ be the edge in the right part that contains the midpoint of that part, where $u_L$ and $u_R$ are closest to $e$ (see Figure~\ref{fig:TopStabRotateUB}b). Furthermore, let $Z$ be the length of $C\setminus\{e, e', e_L, e_R\}$. Now consider the potential edges $(u, v_R)$, $(v, v_L)$, $(u', u_R)$, and $(v', u_L)$. By the triangle inequality, the sum of the lengths of these edges is at most $4|e| + 2|e_L| + 2|e_R| + Z$. Thus, one of these potential edges has length at most $|e| + |e_L|/2 + |e_R|/2 + Z/4$. Without loss of generality let $(u, v_R)$ be that edge (the construction is fully symmetric). We can now rotate $e$ to $(u, v_R)$, then to $(u', v_R)$, and finally to $e'$. As each part of $C$ has length at most $2\OPT/3$, we get that $|(u', v_R)| \leq \OPT/3 + |e|$ by construction. Furthermore we have that $\OPT = |e| + |e_L| + |e_R| + Z$. Thus, $|(u, v_R)| \leq |e| + |e_L|/2 + |e_R|/2 + Z/4 = \OPT/3 + 2|e|/3 + |e_L|/6 + |e_R|/6 - Z/12$. Since $e$ needs to be removed to update the EMST, it must be the longest edge in $C$. Therefore $|(u, v_R)| \leq \OPT/3 + |e|$, which shows that $|x| \leq \OPT/3 + |e|$. Since the length of the intermediate tree is $\OPT - |e| + |x| \leq \frac{4}{3} \OPT$, we obtain that $\TS(\text{EMST}, \mathcal{T}_\mathcal{I}, \mathcal{T}_\mathcal{S}) \leq \frac{4}{3}$.
\end{proof}

\begin{lemma}\label{lem:nrotatelower}
	If $\mathcal{T}_\mathcal{S}$ is defined by edge rotations, then $\TS(\text{EMST}, \mathcal{T}_\mathcal{I}, \mathcal{T}_\mathcal{S}) \geq \frac{10-2\sqrt{2}}{9-2\sqrt{2}} \approx 1.162$.
\end{lemma}
\begin{proof}
	Consider a point in time where the EMST has to be updated by removing an edge $e$ and inserting an edge $e'$. Let the remaining points be arranged in a diamond shape as shown in Figure~\ref{fig:nrotatelower}, where the side length of the diamond is $2$, and $|e| = |e'| = 1$. As a result, the distance between an endpoint of $e$ and the left or right corner of the diamond is $2 - \frac{1}{2}\sqrt{2}$. Now we define a \emph{top-connector} as an edge that intersects the vertical diagonal of the diamond, but is completely above the horizontal diagonal of the diamond. A \emph{bottom-connector} is defined analogously, but must be completely below the horizontal diagonal. Finally, a \emph{cross-connector} is an edge that crosses or touches both diagonals of the diamond. Note that a cross-connector has length at least $2$, and a top- or bottom-connector has length at least $|e| = 1$. In the considered update, we start with a top-connector and end with a bottom-connector. Since we cannot rotate from a top-connector to a bottom-connector in one step, we must reach a state that either has both a top-connector and a bottom-connector, or a single cross-connector. In both options the length of the spanning tree is equal to four times $2 - \frac{1}{2}\sqrt{2}$, plus the connector(s) of length at least $2$. This gives a total length of $10 - 2 \sqrt{2}$, while the minimum spanning tree has no cross-connector, but a single top- or bottom-connector and hence a total length of $9 - 2 \sqrt{2}$.

    \begin{figure}
	   \centering
	   \includegraphics{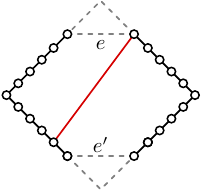}
	   \caption{Lower bound construction for edge rotations.}
	   \label{fig:nrotatelower}
    \end{figure}
	
	To force the update from $e$ to $e'$ in this configuration, we can use the following motion. The points start at the endpoints of $e$ and move with constant speeds to a position where the points are evenly spread around the left and right corner of the diamond. Then the points move with constant speeds to the endpoints of $e'$. The argument above still implies that we need edges of total length at least $2$ intersecting the vertical diagonal of the diamond at some point during the motion. On the other hand, $\OPT \leq 9 - 2 \sqrt{2}$ throughout the motion. Thus $\TS(\text{EMST}, \mathcal{T}_\mathcal{I}, \mathcal{T}_\mathcal{S}) \geq \frac{10-2\sqrt{2}}{9-2\sqrt{2}} \approx 1.162$.
\end{proof}

\begin{theorem}
    For a state-aware algorithm~$\mathcal{A}$ solving the kinetic EMST problem, if $\mathcal{T}_\mathcal{S}$ is defined by edge rotations, then $1.162 \approx \frac{10-2\sqrt{2}}{9-2\sqrt{2}} \leq \TS(\text{EMST}, \mathcal{T}_\mathcal{I}, \mathcal{T}_\mathcal{S}) \leq~\frac{4}{3}$.
\end{theorem}

\subsection{Topological stability of stateless algorithms for EMSTs}\label{sec:topoEMST-stateless}
As we explained in Section~\ref{sec:topostable-analysis}, for stateless algorithms we can prove a lower bound on $\TS$ using the fact that a stateless algorithm is a continuous function. Let $\mathcal{A} \colon \mathcal{I} \rightarrow \mathcal{S}$ be a function between input space~$\mathcal{I}$ and solution space~$\mathcal{S}$. With slight abuse of notation, we often use a set, for example~$D\subseteq\mathcal{I}$ as input for $\mathcal{A}$ to talk about the set~$\mathcal{A}D$ of solutions that algorithm~$\mathcal{A}$ maps the inputs in~$D$ to. Specifically, omitting the brackets ensures that this notation will later correspond to function composition, when we consider (parameterized) paths in~$\mathcal{I}$ as input. In this section, we show that such a function $\mathcal{A}$, that computes a topologically stable (approximation of an) EMST on $n$ moving points, produces at least a $\sqrt{n}$-approximation, if $\mathcal{A}$ is a continuous function. In other words, we prove that $\TS$ is at least $\Omega(\sqrt{n})$. We do so for a topology~$\mathcal{T_S}$ of $\mathcal{S}$ defined by edge flips: any (combinatorial) edge may be replaced by any other edge, as long as all input points are still connected after the flip. This result heavily contrasts our finding for state-aware algorithms: when $\mathcal{T_S}$ is defined by edge flips, $\TS$ is equal to~$1$ for state-aware algorithms, as proved in Theorem~\ref{thm:state-aware-flips}. Besides edge flips, we also shortly consider other operations to update spanning trees at the end of this section. 

%Note that a flip graph defined for edge slides or edge rotations, which we considered earlier, has only a subset of the edges of the flip graph~$G$ for edge flips. Our lower bound in this section extends to all solution spaces defined by connected subgraphs of~$G$, since the topological stability ratio~$\TS$ cannot decrease by removing edges from the flip graph defining the solution space~$\mathcal{S}$.

\begin{theorem}\label{thm:topological-stateless-LB}
    For a stateless algorithm~$\mathcal{A}$ solving the kinetic EMST problem on $n$ (moving) input points, if $\mathcal{T}_\mathcal{S}$ is defined by edge flips, then $\TS(\text{EMST}, \mathcal{T}_\mathcal{I}, \mathcal{T}_\mathcal{S}) = \Omega(\sqrt{n})$.
\end{theorem}

\noindent The remainder of this section works towards proving this theorem. Intuitively, we use topological arguments to show that the continuity of~$\mathcal{A}$ together with the topology of the input and solution spaces, prevents function~$\mathcal{A}$ from being injective. We carefully choose the geometry of the input instances we consider, such that any spanning tree chosen by~$\mathcal{A}$ for the inputs that violate injectivity, results in an $\Omega(\sqrt{n})$ approximation factor.

More specifically, the proof is structured as follows. We assume that~$\mathcal{A}$ is a continuous function and consider a part~$D\subseteq \mathcal{I}$ of the input space that is homeomorphic to a disk. We show that the boundary $\delta D$ of this set $D$ must map to a part of $\mathcal{S}$ that cannot have holes, since~$\mathcal{A}$ is continuous. More precisely, the image of $\delta D$ under $\mathcal{A}$ cannot contain cycles of the flip graph~$\mathcal{S}$ and hence $\delta D$ maps to a tree-like part of~$\mathcal{S}$. Thus, different parts of~$\delta D$ map to the exact same vertex of~$\mathcal{S}$. For those inputs we show an approximation factor of~$\Omega(\sqrt{n})$.

For ease of explanation, we consider inputs on point set~$P = \{p_1, p_2, \ldots, p_{2n}\}$ consisting of $2n$ moving points in $1$-dimensional Euclidean space. This affects our computations only by constant factors compared to inputs with $n$ points, and hence we still work towards proving Theorem~\ref{thm:topological-stateless-LB}. This also means that we use the standard topology on $\reals^{2n}$ as $\mathcal{T_I}$ and $\mathcal{T_S}$ is defined by edge flips, albeit for spanning trees on $2n$ points. However, as will become clear throughout this section, the construction can be embedded in any Euclidean space of higher dimension, and hence the results hold for $d$ dimensions.

\paragraph{Constructing the inputs.} We consider a restricted set of inputs, in which the points in~$P$ lie on a single line, and form two sets~$P_1 = \{p_1, \ldots, p_n\}$ and $P_2 = \{p_{n+1}, \ldots, p_{2n}\}$ of $n$ points each. Figure~\ref{fig:topological-stateless-instance} provides illustrations for the upcoming paragraphs. %The points in $P_1$ are located at least a distance $n$ from each point in $P_2$.

\begin{figure}
	\centering
	\includegraphics{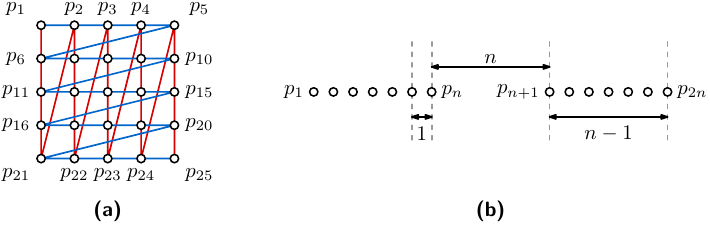}
	\caption{Construction in instances for Theorem~\ref{thm:topological-stateless-LB}: \textbf{\textsf{(a)}} The orderings of the two permutations $\Pi$ and $\Pi'$ on a small point set with $n=25$ (in blue and red, respectively). \textbf{\textsf{(b)}} A concrete input instance with point sets $P_1$ and $P_2$ uniformly spaced, and distance $n$ between the sets.}
	\label{fig:topological-stateless-instance}
\end{figure}

We are going to consider two permutations $\Pi$ and $\Pi'$ of both $P_1$ and $P_2$; either point set can be in either permutation. Let $p_i,\ldots, p_{i+n-1}$ be the points in~$P_1$ ($i=1)$ or $P_2$ ($i=n+1$). Permutation $\Pi$ simply orders the points on their index. Permutation $\Pi'$ also starts at $p_i$ but consists of $\sqrt{n}$ sequences of length $\sqrt{n}$ that skip $\sqrt{n}$ indices for consecutive points and consecutive sequences shift all indices by~1: 

\begin{align*}
    &p_i, p_{i+\sqrt{n}}, p_{i+2\sqrt{n}}, \ldots, p_{i+n-\sqrt{n}},\\ 
    &p_{i+1}, p_{i+1+\sqrt{n}}, p_{i+1+2\sqrt{n}}, \ldots, p_{i+1+n-\sqrt{n}},\\ 
    &\ldots\\
    &p_{i+\sqrt{n}-1}, p_{i+2\sqrt{n}-1}, p_{i+3\sqrt{n}-1}, \ldots,p_{i+n-1}
\end{align*}

Intuitively, $\Pi$ and $\Pi'$ correspond to a row-by-row and column-by-column traversal in a $\sqrt{n} \times \sqrt{n}$ grid; see Figure~\ref{fig:topological-stateless-instance}a for a visual interpretation. Observe that both permutations have $p_i$ as the first point and $p_{i+n-1}$ as the last point; this will be helpful later.

To create concrete input instances, we consider combinations $X \Join Y$ of permutations $X$ and $Y$ of the subsets $P_1$ and $P_2$, respectively: For such an instance $X \Join Y$, the points in $P_1$ are ordered according to~$X$, and consecutive points are located at distance~$1$ from one another. Similarly, the points in $P_2$ are ordered according to~$Y$, and consecutive points are also at distance~$1$. The points of~$P_1$ are placed before the points on $P_2$ and the (last) point~$p_n\in P_1$ is at distance~$n$ from the (first) point~$p_{n+1}\in P_2$. See Figure~\ref{fig:topological-stateless-instance}b.

We are now ready to construct our set~$D\subseteq \mathcal{I}$ of inputs; see Figure~\ref{fig:topological-stateless-inputspace}a. First, we define the instances on the boundary~$\delta D$ of~$D$: Consider the four instances $I_a = \Pi \Join \Pi$, $I_b = \Pi \Join \Pi'$, $I_c = \Pi' \Join \Pi'$, and $I_d = \Pi' \Join \Pi$. The boundary $\delta D$ consists of four paths $\delta_{ab},\delta_{bc},\delta_{cd}$ and $\delta_{da}$ through $\mathcal{I}$. The path~$\delta_{ab}$ has its endpoints at $I_a$ and $I_b$, and is otherwise defined by a linear interpolation of the positions of points in $P_2$. Paths $\delta_{bc},\delta_{cd}$ and $\delta_{da}$ are defined analogously: they respectively have endpoints in $I_b$ and $I_c$, $I_c$ and $I_d$, and $I_d$ and $I_a$, and interpolate positions of only $P_1$, only $P_2$, and only $P_1$, respectively.

Observe that the boundary~$\delta D$ basically forms a square in~$\mathcal{I}$. Any point inside the square defines an instance in which the position of the points in $P_1$ and $P_2$ is determined by the interpolated state between $\Pi$ and $\Pi'$: for $P_1$ this state can be found by (orthogonal) projection to $\delta_{bc}$ or $\delta_{da}$ and for $P_2$ by projection to $\delta_{ab}$ or $\delta_{cd}$. These are all instances that make up~$D$. Notice that $D$ forms a compact (that is, closed and bounded) subset of~$\reals^{2dn}$. 

\begin{figure}
	\centering
	\includegraphics{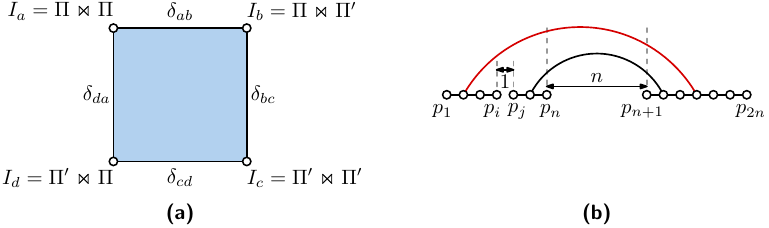}
	\caption{\textbf{\textsf{(a)}} The set of inputs~$D\subseteq\mathcal{I}$. \textbf{\textsf{(b)}} An instance $I_a,I_b,I_c$ or $I_d$ with multiple edges between $P_1$ and $P_2$. These edges are drawn as arcs for clarity, but are straight lines in a geometric spanning tree. Replace the red edge by $(p_i,p_j)$ to shorten the spanning tree.}
	\label{fig:topological-stateless-inputspace}
\end{figure}

We now prove two crucial lemmata on the length of a spanning tree on instances $I_a,I_b,I_c$ and $I_d$, and on just $P_1$ or $P_2$, if they are positioned according to $\Pi$ or $\Pi'$.

\begin{lemma}\label{lem:topological-stateless-crossedge}
    For instances $I_a, I_b, I_c$ and $I_d$, any spanning tree~$T$ that has more than one edge $(p,p')$ with $p\in P_1$ and $p'\in P_2$, can be turned into a spanning tree~$T'$ with only one such edge, such that the total edge length of~$T'$ is smaller than the total edge length of~$T$.
\end{lemma}
\begin{proof}
    Consider such a spanning tree~$T$ (see Figure~\ref{fig:topological-stateless-inputspace}b) and remove one of the edges between~$P_1$ and~$P_2$, which has length at least~$n$.
    % (remove the longest edge if such an edge exists). %% I don't see why this matters
    As a result, the edges no longer form a spanning tree and there are two connected components. 
    There must be a pair of consecutive points that belong to different components, and that both lie in $P_1$ or both in $P_2$; if both sets each form a single component, then we removed the only edge between $P_1$ and $P_2$.
    By adding the edge of length 1 between these two consecutive points, we construct a spanning tree $T'$ of length strictly smaller than $T$.
    We can repeat this modification until only a single edge between $P_1$ and $P_2$ remains, proving the statement.
    %
    %% OLD PROOF
    %Observe that the points in~$P_1$ or the points in~$P_2$ (or both) are separated over the two components. 
    %
    %Assume without loss of generality that the points in~$P_1$ are separated over different connected components. There must be a pair $p_i,p_j$ of consecutive points in~$P_1$ such that $p_i$ is in one components and $p_j$ in the other component; if not, then all points are in the same component. By adding the edge $(p_i,p_j)$, of length $1$, we reconnect the components and have constructed a spanning tree~$T'$ with smaller total length than~$T$. We can repeat this modification until only a single edge between $P_1$ and $P_2$ remains, proving the statement.
\end{proof}

\begin{lemma}\label{lem:topological-stateless-permutations}
   For a set~$P$ of $n$ uniformly-spaced collinear points, and an arbitrary (combinatorial) spanning tree~$T$ on $P$, the length of $T$ is an $\Omega(\sqrt{n})$-approximation of the length of an EMST, when the permutation of the points in $P$ is $\Pi$ or $\Pi'$.
\end{lemma}
\begin{proof}
    Consider an arbitrary spanning tree~$T$ on such a set of $n$ points. We derive a lower bound on the length of $T$ when the points are permuted according to $\Pi$ or $\Pi'$. %For each edge of $T$ we consider its length when the points follow one of the two permutations. We show below that for one of the two permutations, the length is at least $\Omega(\sqrt{n})$. Therefore, every edge has length $\Omega(\sqrt{n})$ for either~$\Pi$ or $\Pi'$. 
%
%   To prove this bound, we first prove that every edge has length $\Omega(\sqrt{n})$ in $\Pi$ or in $\Pi'$.
    
    We show that, for $\Pi$ or $\Pi'$, at least half of the edges of $T$ are of length $\Omega(\sqrt{n})$, and hence $T$ has a total length of at least~$\Omega(n\sqrt{n})$ for that permutation. Observe that an EMST connects the points in the order in which they occur on the line, leading to a length of $\OPT = n-1$. As such, the length of $T$ is $\Omega(\sqrt{n})\cdot \OPT$ for one of the two permutations. We complete the proof by proving that an arbitrary edge of~$T$ has length $\Omega(\sqrt{n})$ for $\Pi$ or $\Pi'$.

    Let $e=(p_i,p_j)$ be an arbitrary edge of~$T$, with $i<j$. If $j-i\geq\sqrt{n}$ then for $\Pi$ points $p_i$ and $p_j$ are at least a distance $\sqrt{n}$ apart: the points are in order of their index and uniformly distributed over a length-$n$ line piece, putting consecutive points at distance $1$. 
    
    If $j-i<\sqrt{n}$, consider the distance between $p_i$ and $p_j$ for $\Pi'$. Remember that~$\Pi'$ consists of $\sqrt{n}$ sequences of length~$\sqrt{n}$ in which $\sqrt{n}$ indices are skipped for consecutive points. Therefore, $p_i$ and $p_j$ must be part of different sequences. Furthermore, if $p_i$ and $p_j$ are not part of consecutive sequences, there are at least $\sqrt{n}$ points between them in the sequence, and hence their distance is at least $\sqrt{n}$. Thus, finally, consider the case where $p_i$ and $p_j$ are part of consecutive sequences of $\Pi'$. Then either $j = i+1$ and $p_i$ and $p_j$ are exactly at distance $\sqrt{n}$, or $j=i-1+\sqrt{n}$ and $p_i$ and $p_j$ are exactly at distance $\sqrt{n}-1 = \Omega(\sqrt{n})$.
\end{proof}

\paragraph{Mapping to the flip graph.}
Given our input instances~$D\subseteq \mathcal{I}$ we can now investigate their mapping by~$\mathcal{A}$ to the flip graph defining~$\mathcal{S}$. We start considering specific instances in~$\delta D$.%We start by proving that Theorem~\ref{thm:topological-stateless-LB} holds when~$\mathcal{A}$ violates injectivity for any pair of instances out of $I_a, I_b, I_c$ and $I_d$.

\begin{lemma}\label{lem:topological-stateless-approx}
    Let $I,I'\in \delta D$, with $I\neq I'$, be two instances for which $P_1$ or $P_2$ are positioned as follows: the points are collinear, have unit distance between consecutive pairs, and are ordered according to $\Pi$ in $I$ and according to $\Pi'$ in $I'$. If $I$ and $I'$ are mapped by~$\mathcal{A}$ such that $\mathcal{A}(I)=\mathcal{A}(I')$, then $\mathcal{A}(I)$ or $\mathcal{A}(I')$ is an $\Omega(\sqrt{n})$-approximation of the EMST.
\end{lemma}
\begin{proof}
    First consider the spanning trees $\mathcal{A}(I)$ and $\mathcal{A}(I')$. These are spanning trees on $2n$ points, divided over two $n$-point subsets $P_1$ and $P_2$. By Lemma~\ref{lem:topological-stateless-crossedge} we know that for a spanning tree with more than one edge between $P_1$ and $P_2$ there is a tree with only a single such edge, that has smaller total edge length. Since we are proving a lower bound on the total edge length, we may assume that $\mathcal{A}(I)$ and $\mathcal{A}(I')$ have been modified to trees with only a single edge between~$P_1$ and $P_2$, which can only shorten their total edge length.
    
    Assume without loss of generality that the conditions in the lemma hold for~$P_1$. Additionally, as there is only a single edge between $P_1$ and $P_2$, the spanning trees $\mathcal{A}(I)$ and $\mathcal{A}(I')$ consist of a spanning tree on $P_1$ and a spanning tree on $P_2$ that are connected by an edge. As such, we can apply Lemma~\ref{lem:topological-stateless-permutations} to~$P_1$: since $\mathcal{A}(I)=\mathcal{A}(I')$, the same spanning tree is used for both instances, while the permutation differs for $P_1$. We apply Lemma~\ref{lem:topological-stateless-permutations} to obtain that the subtree on $P_1$ is an $\Omega(\sqrt{n})$-approximation of the EMST on $P_1$ for $\mathcal{A}(I)$ or $\mathcal{A}(I')$.

    To finish the proof, we show that even if the remainders of $\mathcal{A}(I)$ and $\mathcal{A}(I')$ are as short as possible, one of them (in its entirety) is still a $\Omega(\sqrt{n})$-approximation. First, the edge between $P_1$ and $P_2$ has length at least $n$. Second, note that $I,I'\in \delta D$ and on the paths $\delta_{ab},\ldots,\delta_{da}$ we always linearly interpolate $P_1$ or $P_2$ between uniformly spaced orderings according to $\Pi$ or $\Pi'$. Since both orderings start and end with the same point, the outer points in $P_2$ are always at a distance of $n-1$. Thus the spanning tree on $P_2$ (and also $P_1$) has length at least $n-1$ and the total length of $\mathcal{A}(I)$ or $\mathcal{A}(I')$ is at least $2n-2 +c\cdot n\sqrt{n}$, for some constant~$c>0$. Similarly, an EMST on $I$ or $I'$ has length at most $2n-2+n=3n-2$, resulting in an approximation factor of
    \begin{equation*}
        \frac{2n-1 + c\cdot n\sqrt{n}}{3n-2} \geq \frac{2n -1 + c\cdot n\sqrt{n}}{3n} = \frac{2}{3} - \frac{1}{3n}+ \frac{c\cdot \sqrt{n}}{3} \geq x \sqrt{n} \quad\text{ for } x= \frac{c}{3} \text{ \& } n\geq \frac{1}{2}.
    \end{equation*}
    Thus, $\mathcal{A}(I)$ or $\mathcal{A}(I')$ is an $\Omega(\sqrt{n})$-approximation of the EMST.
\end{proof}

As Lemma~\ref{lem:topological-stateless-approx} applies to any pair out of $I_a, I_b, I_c$ and $I_d$, we get the following corollary.

\begin{corollary}\label{cor:topological-stateless-boundary}
    If instances $I,I'\in \{I_a, I_b, I_c,I_d\}$, with $I\neq I'$, are mapped by~$\mathcal{A}$ such that $\mathcal{A}(I)=\mathcal{A}(I')$, then $\mathcal{A}(I)$ or $\mathcal{A}(I')$ is an $\Omega(\sqrt{n})$-approximation of the EMST.
\end{corollary}

Finally, we consider the whole boundary~$\delta D$ and prove that either $\mathcal{A}$ violates injectivity for parts of $\delta D$, or $\mathcal{A}$ maps $\delta D$ around a hole in~$\mathcal{S}$ (a cycle in the flip graph), in which case we find a contradiction with $\mathcal{A}$ being a continuous function. For the latter case of this proof, we consider the \emph{fundamental group} of~$\mathcal{S}$. While the fundamental group of a graph is a basic concept in computational topology,for completeness, we introduce the essential concepts of this fundamental group in the next few paragraphs. For additional background information see an algebraic topology textbook, for example Chapter 1 of~\cite{hatcher2002algebraic}.

\paragraph{Fundamental group of a graph.} The fundamental group is a group --- in group theoretical sense --- of the homotopy equivalence classes of paths starting and ending at a base point~$x_0$ in a topological space; these paths are called \emph{loops} based at~$x_0$. For (path-)connected topological spaces, the fundamental group is not dependent on the base, that is, fundamental groups for different base points are equivalent. Since $\mathcal{T}_\mathcal{S}$ is defined by edge flips, the flip graph defining~$\mathcal{S}$ is connected, and hence we may consider the fundamental group with loops based at an arbitrary point in $s\in\mathcal{S}$. 

We describe loops in terms of the edges that are traversed. In such a description, each edge~$e\in S$ is an interval $\gamma_e \colon [0,1] \rightarrow S$, and when two adjacent edges $e$ and $e'$ are traversed consecutively, we write $\gamma_e \cdot \gamma_{e'}$ to describe the concatenation of the paths over these edges. If an edge~$e$ is traversed in the direction in which the parameterization is descending, then we use $\gamma'_e(t) = \gamma_e(1-t)$ to acknowledge this. Partial traversal of edges, for example traversing edge $e$ up to $t\in[0,1]$ and returning back to the endpoint we started from, can be described as $\gamma_e[0,t]\cdot\gamma'_e[1-t,1]$. We need one additional property of the fundamental group to relate these descriptions of loops to homotopy classes.
    
The fundamental group is a \emph{free group}, and therefore every element of the group can be described as the product ($\cdot$) of so called \emph{generators}. Such a description is called a \emph{word}, that, intuitively, corresponds to a concatenation of paths/loops. For topological spaces based on graphs, the fundamental group and its generators can be defined as follows. Consider~$\mathcal{S}$ as a graph with spanning tree $T\subseteq\mathcal{S}$ on the vertices of~$S$. The edges $\mathcal{S} - T$ not in the spanning tree each form one of the generators of the fundamental group. For each edge $e\in S - T$ we hence have an element $\gamma_e$ (and its inverse $\gamma'_e$) in the fundamental group.
    
It is important that our base point $s\in\mathcal{S}$ lies on the tree $T$ defining the generators, that is, $s\in T$. When $s$ lies on an edge~$e_s$ of $S$, we contract $e_s$ and compute a spanning tree on the resulting graph. We can then set $T$ to be the edges of the computed spanning tree, together with~$e_s$, to ensure that $s\in T$. Notice that there is a unique path through $T$ that connects any two endpoints of generators, or a generator endpoint and base point $s$. Each such path along $T$ can be considered an identity element of the fundamental group. A loop in $\mathcal{S}$ based at $s$ (an element of the fundamental group) can hence be described as an alternating sequence of (possibly empty) paths through $T$ (identity elements) and traversals of edges in $S-T$ (generator elements). %In the next paragraph we elaborate more on the relation between words and loops.

Conversely, a word can be translated to a loop: given a word, start at base point $s\in S$, take the (unique) path along the tree $T$ defining the generators, from $s$ to the endpoint of the first generator edge, traverse the generator edge, follow $T$ to the next generator edge, and so on. After the last generator, the loop is finished by returning to base point $s$ over $T$. Note that for a generator $e$, either $\gamma_e$ or $\gamma'_e$ may occur in word, and in the corresponding loop the edge $e$ is traversed in different directions for $\gamma_e$ and $\gamma'_e$. By definition of a free (fundamental) group, each equivalence class (of homotopic loops based at $s$) can be represented by a unique representative word; intuitively this is a ``shortest'' loop in the homotopy class.

\begin{proof}[Proof of Theorem~\ref{thm:topological-stateless-LB}]
    Consider the boundary~$\delta D$ of our input instances~$D$. This boundary contains instances $I_a, I_b, I_c$ and $I_d$, which $\mathcal{A}$ maps to different solutions in~$\mathcal{S}$ or the theorem holds by Corollary~\ref{cor:topological-stateless-boundary}. Hence, assume that $I_a,I_b,I_c,I_d$ map to different solutions in~$\mathcal{S}$.

    Remember, $\mathcal{T_S}$ is defined by a flip graph, and we consider~$\mathcal{S}$ to be the corresponding simplicial $1$-complex. We make a case distinction on whether certain paths of $\delta D$ are mapped injectively to~$\mathcal{S}$ or not. %In case $\delta D$ is mapped to a tree~$T_\delta\subseteq \mathcal{S}$
    The instances $I_a, I_b, I_c$ and $I_d$ are all mapped to different points in~$\mathcal{S}$ and the (therefore non-empty) paths $\delta_{ab},\delta_{bc},\delta_{cd},\delta_{da}\subseteq \delta D$ connect the respective instances. %via~$T_\delta$. 
    We first consider the cases where $\mathcal{A}\delta_{ab}$ and $\mathcal{A}\delta_{cd}$ share a vertex, and where $\mathcal{A}\delta_{bc}$ and $\mathcal{A}\delta_{da}$ share a vertex. We end with the case where neither of the first two cases apply, and hence the paths form at least a 4-cycle; see Figure~\ref{fig:topological-stateless-mapcases} for a visualization of these cases. 
    
    %This tree~$T_\delta$, corresponding to $\mathcal{A}\delta D$ has the following properties: From each vertex corresponding to $\mathcal{A}(I_a), \mathcal{A}(I_b), \mathcal{A}(I_c)$ or $\mathcal{A}(I_d)$, there is a (possibly empty) path to a vertex~$v_x\in~\mathcal{S}$. Without loss of generality consider $\mathcal{A}(I_a)$. In case the path is empty, we have $v_x=\mathcal{A}(I_a)$. On the path from $\mathcal{A}(I_a)$ to $v_x$ the paths $\delta_{ab}$ towards $\mathcal{A}(I_b)$ and $\delta_{da}$ towards $\mathcal{A}(I_d)$ must be mapped. 
    %Since $\mathcal{A}(I_b)$ and $\mathcal{A}(I_d)$ are mapped to different vertices of the flip graph, $T_\delta$ splits into (at least) two branches at $v_x$; one subtree $T_b\subseteq T_\delta$ containing $\mathcal{A}(I_b)$ and one subtree $T_d\subseteq T_\delta$ containing $\mathcal{A}(I_d)$. We make a case distinction on whether $T_b$ or $T_d$ contains $\mathcal{A}(I_c)$; see Figures~\ref{fig:}a and \ref{fig:}b, respectively. Note that in a degenerate case there may be a separate branch connecting at $v_x$ containing $\mathcal{A}(I_c)$. We can simply consider this branch part of $T_b$ or $T_d$, by considering a subtree rooted at $v_x$ with more than one child; see Figure~\ref{fig:} for an illustration of this case in which $T_b$ is that larger subtree.

    \begin{figure}
        \centering
        \includegraphics{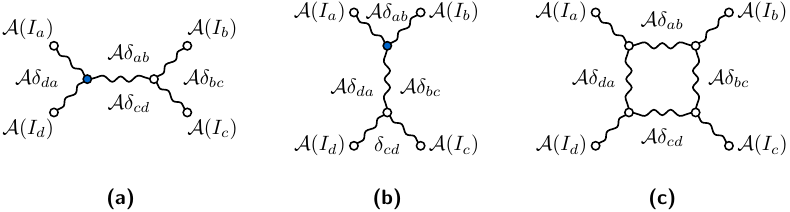}
	   \caption{Cases for the mapping of $\delta D$ produced by~$\mathcal{A}$. Opposite paths share a vertex, for example the blue marked vertex, \textbf{\textsf{(a)}}~on $\mathcal{A}\delta_{ab}$ and $\mathcal{A}\delta_{cd}$, or \textbf{\textsf{(b)}} on $\mathcal{A}\delta_{bc}$ and $\mathcal{A}\delta_{da}$. Otherwise, \textbf{\textsf{(c)}} $\mathcal{A}\delta_{ab} \cap \mathcal{A}\delta_{cd} = \emptyset$ and $\mathcal{A}\delta_{bc} \cap \mathcal{A}\delta_{da} = \emptyset$.}
	   \label{fig:topological-stateless-mapcases}
    \end{figure}
    
    Consider the case where $\mathcal{A}\delta_{ab}$ and $\mathcal{A}\delta_{cd}$ share a vertex; see Figure~\ref{fig:topological-stateless-mapcases}a. 
    % $T_b$ contains $\mathcal{A}(I_c)$. Observe that the path $\delta_{ab}$ is mapped in~$\mathcal{S}$ to a path from $\mathcal{A}(I_a)$ through $v_x$ to $\mathcal{A}(I_b)$ via a path in $T_b$. Similarly, path $\delta_{cd}$ is mapped in~$\mathcal{S}$ to a path from $\mathcal{A}(I_d)$ to $v_x$ via a path in $T_d$ and then proceeds to $\mathcal{A}(I_c)$ via $T_b$. Crucially, both paths coincide in~$v_x$. 
    Remember that the instances on $\delta_{ab}$ require the points in $P_1$ to be ordered according to $\Pi$ while the instances on $\delta_{cd}$ require those points to be ordered according to~$\Pi'$. Additionally, the points in $P_1$ are not interpolating, and hence are uniformly spaced. Since the shared vertex provides a spanning tree for both these orderings, we can apply Lemma~\ref{lem:topological-stateless-approx} to show that an instance on $\delta_{ab}$ or an instance on $\delta_{cd}$ is mapped to an $\Omega(\sqrt{n})$-approximation of the EMST.
    The case in which $\mathcal{A}\delta_{bc}$ and $\mathcal{A}\delta_{da}$ share a vertex (see Figure~\ref{fig:topological-stateless-mapcases}b)
    %$T_d$ contains $\mathcal{A}(I_c)$ 
    is analogous: 
    %We now consider paths $\delta_{bc}$ and $\delta_{da}$, which coincide at~$v_x$ and
    these paths disagree on the permutation of the (static) point set~$P_2$ and again Lemma~\ref{lem:topological-stateless-approx} applies.

    Finally, we consider the case in which both $\mathcal{A}\delta_{ab} \cap \mathcal{A}\delta_{cd} = \emptyset$ and $\mathcal{A}\delta_{bc} \cap \mathcal{A}\delta_{da} = \emptyset$. We use a topological argument to derive a contradiction in this case: We consider two paths $f$ and $g$ starting at some point $u \in \delta_{da}$ and ending at some point $v\in \delta_{bc}$; see Figure~\ref{fig:topological-stateless-homotopicpaths}a. These path go around $D$ on different sides, and we use the fact that $f$ and $g$ are homotopic. Specifically, we define $f,g \colon [0,1] \rightarrow D$ with $f(0)=I_a = u$, and the path defining $f$ then follows $\delta_{ab}$ to $I_b$, and finally reaches $v=I_c$ via $\delta_{bc}$, such that $f(1)=v$. Similarly, $g(0)=I_a =u$, but $g$ follows $\delta_{da}$ to $I_d$, then and reaches $v=I_c$ via $\delta_{cd}$, with $g(1)=v$. Observe that $f$ and $g$ are indeed homotopic in $\mathcal{D}$ relative to the endpoints $u$ and $v$, since we can linearly interpolate between them without leaving~$D$. Thus, since $\mathcal{A}$ is a continuous function, $\mathcal{A}f$ and $\mathcal{A}g$ are homotopic continuous paths in $\mathcal{S}$ relative to endpoints $\mathcal{A}(u)$ and $\mathcal{A}(v)$.

    To eventually derive a contradiction, we consider the fundamental group of~$\mathcal{S}$; specifically the fundamental group with base point $\mathcal{A}(v)$. For our proof we work towards finding representative words for the homotopy classes of loops containing $\mathcal{A}f$ or $\mathcal{A}g$. We now show how to sensibly turn $\mathcal{A}f$ and $\mathcal{A}g$ into loops $\ell_f$ and $\ell_g$, respectively. %, for which we can determine the homotopy class via the fundamental group.
    For technical reasons, we add the edge $e^* = (\mathcal{A}(u),\mathcal{A}(u))$ to $\mathcal{S}$; the exact reason will become clear in the next paragraphs. Observe that, before adding~$e^*$, the tree $T$ defining the generators of the fundamental group is a spanning tree for $\mathcal{S}$, and still is after adding $e^*$. Hence we simply added a generator element to the fundamental group in this way. We now define $\ell_f$ as traversing the path $\textsc{path}_T(\mathcal{A}(v),\mathcal{A}(u))$ along $T$ from $\mathcal{A}(v)$ to $\mathcal{A}(u)$, concatenating this with $\gamma_{e^*}$, and finally traversing $\mathcal{A}f$ (from $\mathcal{A}(u)$ back to $\mathcal{A}(v)$). We define $\ell_g$ analogously. With slight abuse of notation, we interchangeably refer to $\ell_f$ and $\ell_g$ (and other loops) as loops and their words. Our goal is now to find representative words for $\ell_f$ and $\ell_g$.

    Remember that by Corollary~\ref{cor:topological-stateless-boundary}, we have $\mathcal{A}(u) = \mathcal{A}(I_a) \neq \mathcal{A}(I_c)=\mathcal{A}(v)$. Hence, both $\mathcal{A}f$ and $\mathcal{A}g$ must be non-trivial, that is, they are not just a point. We define a deformation retract of both $\mathcal{A}f$ and $\mathcal{A}g$ that removes any ``unnecessary'' spikes from the images of $f$ and $g$: if at any point $(\mathcal{A}f)(x)$, we have that $(\mathcal{A}f)(x-\varepsilon)=(\mathcal{A}f)(x+\varepsilon)$, for arbitrarily small~$\varepsilon$, then we can slide $(\mathcal{A}f)(x)$ in the direction of $(\mathcal{A}f)(x-\varepsilon)$ to shorten the path; see Figure~\ref{fig:topological-stateless-homotopicpaths}b. If this happens on an edge of the tree $T$ defining the fundamental group, then this does not affect the word $\ell_f$. However, when spike removal happens on a generator, this corresponds to (partially) traversing a generator edge $e$, followed by traversing the same edge in opposite direction. Specifically, $\ell_f$ then has a subsequence $...\gamma_e\cdot\gamma'_e ...$ of \emph{adjacent inverses} (or for partial edge traversal a subsequence $...\gamma_e[0,t]\cdot\gamma'_e[1-t,1] ...$) that we would like to remove. The deformation retract defined above hence corresponds to the removal of adjacent inverses from the sequence of generators describing~$\ell_f$. 

    \begin{figure}
        \centering
        \includegraphics[page=1]{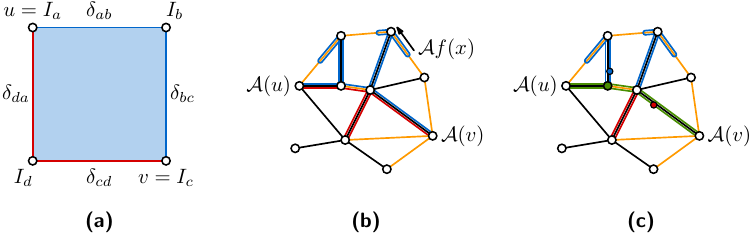}
	   \caption{\textbf{\textsf{(a)}} Homotopic paths $f$ (in blue) and $g$ (in red). \textbf{\textsf{(b)}} The mappings $\mathcal{A}f$ (in blue) and $\mathcal{A}g$ (in red). The edges of the graph $S$ are partitioned into spanning tree $T$ (in black) and generator edges (in yellow). The point $(\mathcal{A}f)(x)$ can be retracted. \textbf{\textsf{(c)}} Base paths $F=G$ are highlighted in green. The blue vertex shows $\mathcal{A}(I_b)$, while red vertex shows $\mathcal{A}(I_d)$. Since $\mathcal{A}(I_b)\notin F$, the vertex on the base path closest to $\mathcal{A}(I_b)$ is drawn in green.}
	   \label{fig:topological-stateless-homotopicpaths}
    \end{figure}
    
    We want to repeatedly remove spikes from $\mathcal{A}f$ and $\mathcal{A}g$ until we find the \emph{base paths} $F,G\colon [0,1] \rightarrow \mathcal{S}$, respectively, that connect $\mathcal{A}(u)$ to $\mathcal{A}(v)$ and cannot be shortened anymore using this operation. Equivalently, we can repeatedly check for adjacent inverses in the part of the words $\ell_f$ and $\ell_g$, that correspond to $\mathcal{A}f$ and $\mathcal{A}g$, until we find the words $\ell_F$ and $\ell_G$, respectively, which correspond to loops that concatenate $\textsc{path}_T(\mathcal{A}(v),\mathcal{A}(u))$ with a traversal of a base path ($F$ and $G$, respectively) from $\mathcal{A}(u)$ to $\mathcal{A}(v)$. Note that adding the edge~$e^* = (\mathcal{A}(u),\mathcal{A}(u))$ was crucial here: we added $e^*$ later, so it is a generator that was not traversed by $\mathcal{A}f$ or $\mathcal{A}g$. The removal of adjacent inverses hence can never remove $\gamma_{e^*}$, and hence $F$ and $G$ must keep $\mathcal{A}(u)$ and $\mathcal{A}(v)$ as their endpoints. Our removal of adjacent inverses corresponds to a deformation retract, which is a homotopy. Since homotopies are equivalence relations, $\mathcal{A}f$, $\mathcal{A}g$, $F$, and $G$ must be homotopic. Similarly, the homotopy classes of loops $\ell_F$ and $\ell_G$ have not changed compared to $\ell_f$ and $\ell_g$, respectively; in fact removing all adjacent inverses results in the representative words of the respective homotopy classes.

    We are now ready to derive a contradiction. We make a case distinction on whether the words $\ell_F$ and $\ell_G$ are exactly the same or not. 
    %If $F$ and $G$ have a non-empty symmetric difference, then traversing $F$ and $G$ from $u$ to $v$ we find subpaths $F'\subseteq F$ and $G'\subseteq G$ that are disjoint, except for their shared endpoints $u',v'\in \mathcal{S}$. Notice that $F'$ and $G'$ form a cycle, which cannot be a subset of $T$ and hence contains at least one edge that is a generator. This edge is part of either $F'$ or $G'$, but not both, and hence this generator element is part of $\ell_F$ or $\ell_G$, but not both. 
    We consider the case where $\ell_F$ and $\ell_G$ have different words first. By definition of the considered fundamental group being a free group, $\ell_F$ and $\ell_G$ must then be in different homotopy classes. Remember that we removed adjacent inverses from $\ell_f$ and $\ell_g$ to obtain $\ell_F$ and $\ell_G$, and that by definition $\ell_f$ and $\ell_g$ traverse our newly added edge~$e^*$ only once. As such, $\gamma'_{e^*}$ cannot occur in either $\ell_f$ or $\ell_g$, so $\gamma_{e^*}$ must still be part of $\ell_F$ and $\ell_G$ after removing adjacent inverses. %The deformation retract corresponding to removing the adjacent inverses from $\ell_f$ and $\ell_g$ hence corresponds to a continuous deformation of $\mathcal{A}f$ to $F$ and of $\mathcal{A}g$ to $G$; all other parts of $\ell_f$ and $\ell_g$ already coincided (the path from $\mathcal{A}(u)$ to $\mathcal{A}(v)$ via $T$), or cannot have changed (edge $e^*$). In particular, 
    As a consequence, the common endpoints $\mathcal{A}(u)$ and $\mathcal{A}(v)$ of both $\mathcal{A}f$ and $\mathcal{A}g$ are still part of $\ell_F$ and $\ell_G$ and hence are crucially still the endpoints of $F$ and $G$. As $\ell_F$ and $\ell_G$ are in different homotopy classes, and they are equivalent except for subpaths $F$ and $G$, respectively, there cannot be a continuous deformation between $F$ and $G$. Thus $F$ and $G$ cannot be homotopic (with respect to endpoints $\mathcal{A}(u)$ and $\mathcal{A}(v)$), leading to a contradiction.

    Let us finally consider the case in which the words $\ell_F$ and $\ell_G$ are exactly the same, and hence $\mathcal{A}f$ and $\mathcal{A}g$ can both retract to the exact same base path $F=G$. Consider the images $\mathcal{A}(I_b) \in \mathcal{A}f$ and $\mathcal{A}(I_d) \in \mathcal{A}g$ and remember that we assumed that $\mathcal{A}\delta_{ab} \cap \mathcal{A}\delta_{cd} = \emptyset$ and $\mathcal{A}\delta_{bc} \cap \mathcal{A}\delta_{da} = \emptyset$. 
    %These images of $\delta_{ab}$ and $\delta_{cd}$ must be continuous subpaths of $\mathcal{A}f$ and $\mathcal{A}g$, respectively. Let $u_1,v_1\in F$ be two points on the base path $F$ at $\mathcal{A}(I_a)$, or as close as possible before, and at $\mathcal{A}(I_b)$, or as close as possible after, respectively
    Additionally, by Corollary~\ref{cor:topological-stateless-boundary} we know that $\mathcal{A}(I_b)\neq\mathcal{A}(I_d)$. Observe that $\mathcal{A}(I_b)$ and $\mathcal{A}(I_d)$ may not necessarily lie on the base path $F=G$, but they can be on branches that are retracted onto the base path. We therefore consider the points $b$ and $d$ on $F=G$ closest to $\mathcal{A}(I_b)$ and $\mathcal{A}(I_d)$, respectively; see Figure~\ref{fig:topological-stateless-homotopicpaths}c. 
    %Similarly, we define $u_2,v_2\in G$ as two points on base path $G$ at $\mathcal{A}(I_d)$, or closest before, and at $\mathcal{A}(I_c)$, or closest after, respectively. Observe that the pairs $u_1,v_1$ and $u_2,v_2$ may either interleave, nest, or occur in sequence.
    We consider the order in which we encounter $b$ and $d$ when traversing the base path from $\mathcal{A}(u)$ to $\mathcal{A}(v)$: either $b$ may occur before $d$, $d$ may occur before $b$, or $b=d$.
    %However, $u_1$ and $v_1$ cannot be interleaved with $u_2$ and $v_2$ along $F=G$, since that would contradict $\mathcal{A}\delta_{ab} \cap \mathcal{A}\delta_{cd} = \emptyset$. 
    Let $b$ occur before $d$, and consider the interval of the base path between $b$ and $d$. The part of $\mathcal{A}f$ that maps to this interval comes after $\mathcal{A}(I_b)$, and thus is part of $\mathcal{A}\delta_{bc}$. Similarly, a part of $\mathcal{A}g$ before $\mathcal{A}(I_d)$ is mapped to this interval, and is hence part of $\mathcal{A}\delta_{da}$. This contradicts our assumption that $\mathcal{A}\delta_{bc} \cap \mathcal{A}\delta_{da} = \emptyset$. Analogously, we get a contradiction with $\mathcal{A}\delta_{ab} \cap \mathcal{A}\delta_{cb} = \emptyset$, when $d$ occurs before $b$.
    %Additionally, $u_1$ and $v_1$ may be nested between $u_2$ and $v_2$, or vice versa, and assume without loss of generality that the former case holds. When this happens, either some part of $\mathcal{A}\delta_{ab}$ is mapped onto $\mathcal{A}\delta_{cd}$ and we again derive a contradiction, or the entirety of $\mathcal{A}\delta_{ab}$ is not mapped to $F$, in particular $\mathcal{A}(I_a)$ and $\mathcal{A}(I_b)$. Since no part of $\delta_{ab}$ is mapped onto $\mathcal{A}\delta_{cd}$, both $\mathcal{A}(I_a)$ and $\mathcal{A}(I_b)$ must be part of the same branch that is retracted to $F$. Hence, a point on $f$ before $I_a$ maps to $u_1$ and a point on $f$ after $I_b$ maps to $v_1$. These points on $f$ are respectively part of $\delta_{da}$ and $\delta_{bc}$, and since $u_1=v_1$ in this case, we get a contradiction with our other assumption that $\mathcal{A}\delta_{bc} \cap \mathcal{A}\delta_{da} = \emptyset$.

    % Thus, when moving along $F=G$ from $\mathcal{A}(u)$ to $\mathcal{A}(v)$, we either first encounter $u_1$ and $v_1$ and then $u_2$ and $v_2$, or vice versa. In the former case, consider the subpath of $F=G$ between $v_1$ and $u_2$. The subpath of $F$ after $v_1$ must correspond to a subset of $\mathcal{A}\delta_{bc}$, while the subpath of $G$ before $u_2$ corresponds to a subset of $\mathcal{A}\delta_{da}$. As such, $\mathcal{A}\delta_{bc} \cap \mathcal{A}\delta_{da} \neq \emptyset$ contradicting our earlier assumption. The latter case is handled analogously by considering the subpath of $F$ before $u_1$ and the subpath of $G$ after $v_2$.
    Lastly, we also derive a contradiction when $b$ and $d$ coincide. Note that this case can occur only when $\mathcal{A}(I_b)$ or $\mathcal{A}(I_d)$ does not lie on the base path, since $\mathcal{A}(I_b)\neq\mathcal{A}(I_d)$. Even when this happens, it holds that $b$ is located at the intersection of $\mathcal{A}\delta_{ab}$ and $\mathcal{A}\delta_{bc}$. To see this, observe that either $b=\mathcal{A}(I_b)$, or $\mathcal{A}f$ deviates from the base path at $b$ before reaching $\mathcal{A}(I_b)$, and returns to $b$ after passing $\mathcal{A}(I_b)$. Similarly, $d$ lies at the intersection of $\mathcal{A}\delta_{cd}$ and $\mathcal{A}\delta_{da}$. Thus, when $b=d$ there is a point where $\mathcal{A}\delta_{ab}$, $\mathcal{A}\delta_{bc}$, $ \mathcal{A}\delta_{cd}$, and $\mathcal{A}\delta_{da}$ intersect, contradicting both $\mathcal{A}\delta_{ab} \cap \mathcal{A}\delta_{cd} = \emptyset$ and $\mathcal{A}\delta_{bc} \cap \mathcal{A}\delta_{da} = \emptyset$.
\end{proof}

Interestingly, the above proof relies entirely on the geometry of the considered input instances and on the fact that the topology of the solution space resembles a flip graph. However, observe that the exact structure of the flip graph is irrelevant and in fact, the flip graph needs to be only connected. Theorem~\ref{thm:topological-stateless-LB} hence extends to the following corollary.

\begin{corollary}\label{cor:topological-stateless-LB}
    For a stateless algorithm~$\mathcal{A}$ solving the kinetic EMST problem on $n$ (moving) input points, if $\mathcal{T}_\mathcal{S}$ is defined by a connected flip graph, then $\TS(\text{EMST}, \mathcal{T}_\mathcal{I}, \mathcal{T}_\mathcal{S}) = \Omega(\sqrt{n})$.
\end{corollary}

\section{Lipschitz stability}\label{sec:Lipstable}
The major drawback of topological stability analysis is that it still does not fully capture stable behavior; the algorithm must be continuous, but we can still make many changes to the solution in an arbitrarily small time frame. In Lipschitz stability analysis we additionally limit how fast the solution can change.

\subsection{Lipschitz stability analysis}
To formally define Lipschitz stability, we return to the continuous setting in Section~\ref{sec:framework}. Let $\Pi$ be an optimization problem with input instances from the set $\mathcal{I}$, solutions from the set $\mathcal{S}$, and optimization function $f$. An input $I\colon [0, 1] \rightarrow \mathcal{I}$ is a continuous path through input space $\mathcal{I}$ and an algorithm $\mathcal{A}$ maps $I(t) \in \mathcal{I}$ to a solution $S(t) \in \mathcal{S}$ for each time $t \in [0,1]$. To quantify how fast a solution changes as the input changes, we need to specify metrics $d_\mathcal{I}$ and $d_\mathcal{S}$ on $\mathcal{I}$ and $\mathcal{S}$, respectively. An algorithm $\mathcal{A}$ is \emph{$K$-Lipschitz stable} if $\St(\mathcal{A}) \leq K$, where $\St(\mathcal{A})$ is defined as in Equations~\ref{eq:continuousstability} and \ref{eq:continuousalgorithmstability}. For an output $S$ of a $K$-Lipschitz stable algorithm, this means that we bound how quickly $S$ can change relative to input $I$. In particular, if we consider two times $t,t'\in[0,1]$, then for this output $S$ it holds that $\Delta_\mathcal{S}(S(t), S(t')) \leq K \Delta_\mathcal{I}(I(t), I(t'))$, where $\Delta_\mathcal{S}$ and $\Delta_\mathcal{I}$ measure the distance traveled in solution and input space, respectively. Figure~\ref{fig:lipschitz-input-solution-space} gives an overview the above.

\begin{figure}
	\centering
	\includegraphics[page=1]{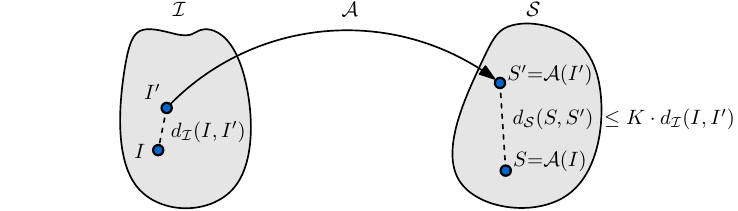}
	\caption{Algorithm $\mathcal{A}$ maps input instances from input space $\mathcal{I}$ to solution space $\mathcal{S}$. The metrics $d_\mathcal{I}$ and $d_\mathcal{S}$ allow us measure the distance traveled in input ($\Delta_\mathcal{I}$) and output ($\Delta_\mathcal{S}$).}
	\label{fig:lipschitz-input-solution-space}
\end{figure}

Note that by defining $K$-Lipschitz stability using Equation~\ref{eq:continuousstability}, we require that, for bounded $K$, a solution $S$ does not change unless input $I$ also changes. However, we can instead use the definition of stability in Equation~\ref{eq:continuousstability2}. Using that definition with the assumption that the input changes with at most unit speed (e.g., points moving with unit speed), allows the solution $S$ to change with speed at most $K$. This assumption is common in computational geometry, and we use it in~\cite{meulemans2019stability} to perform $K$-Lipschitz stability analysis.

Let $\mathcal{P}_\mathcal{I}$ be the set of continuous paths through $\mathcal{I}$. We are again interested in the approximation ratio $\LS$ of any $K$-Lipschitz stable algorithm with respect to $\OPT$. We call this ratio the \emph{Lipschitz stability ratio}, which is defined as
\begin{equation}
\LS(\Pi, K, d_\mathcal{I}, d_\mathcal{S}) = \inf_{\mathcal{A}} \sup_{I \in \mathcal{P}_\mathcal{I}} \sup_{t \in [0,1]} \frac{f(I(t), S(t))}{f(I(t), \OPT(I(t)))}
\end{equation}
where the infimum is taken over all $K$-Lipschitz stable algorithms. It is easy to see that $\LS(\Pi, K, d_\mathcal{I}, d_\mathcal{S})$ is lower bounded by $\TS(\Pi, \mathcal{T}_\mathcal{I}, \mathcal{T}_\mathcal{S})$ for the corresponding topologies $\mathcal{T}_\mathcal{I}$ and $\mathcal{T}_\mathcal{S}$ of $d_\mathcal{I}$ and $d_\mathcal{S}$, respectively.

As mentioned in Section~\ref{sec:framework}, analyses of this type are often difficult. First, we need to be very careful when choosing metrics the $d_\mathcal{I}$ and $d_\mathcal{S}$, as they should behave similarly with respect to scale. For example, let the input consist of a set of points in the plane and let $c I$ for $I \in \mathcal{I}$ be the instance obtained by scaling all coordinates of the points in $I$ by the factor $c$. Now $d_\mathcal{I}$ depends linearly on scale, that is $d_\mathcal{I}(c I, c I') \sim c d_\mathcal{I}(I, I')$. Moreover, assume that $d_\mathcal{S}$ is independent of scale. If we now scale the input instance, we change the relative speed between changes in the input and changes in the solution. Thus without scale invariance, the analysis is rendered meaningless.

Second, we need to be careful with discrete solution spaces. Using the flip graphs as mentioned in Section~\ref{sec:topostable} we can extend a discrete solution space to a continuous space by including the edges. However, it is not always straightforward to extend $d_\mathcal{S}$ over the edges of the flip graph: the distance between (combinatorial) solutions, which are represented by vertices in the flip graph, and between intermediate solutions, on the edges of the flip graph, can change as the input moves. In Section~\ref{sec:lipschitzemst}, we show how to overcome this problem for EMSTs, by carefully defining $d_\mathcal{S}$.

Typically we expect it to be hard to fully describe $\LS(\Pi, K, d_\mathcal{I}, d_\mathcal{S})$ as a function of $K$. However, it may be possible to obtain interesting results for certain values of $K$. One value of interest is the value of $K$ for which the approximation ratio equals or approaches the approximation ratio of the corresponding topological stability analysis. Another potential value of interest is the value of $K$ below which any $K$-Lipschitz stable algorithm performs asymptotically as poorly as a constant algorithm, always computing the same solution regardless of instance.

\subsection{Lipschitz stability of EMSTs}\label{sec:lipschitzemst}
We use the same setting of the kinetic EMST problem as in Section~\ref{sec:topoEMST}, except that, instead of topologies, we specify metrics for $\mathcal{I}$ and $\mathcal{S}$. For $d_\mathcal{I}$ we can simply use the metric in Equation~\ref{eq:Euclmetric}, which implies that points move with a bounded speed. For $d_\mathcal{S}$ we use a metric inspired by the edge slides of Section~\ref{sec:topoEMST}. To that end, we need to define how long a particular edge slide takes, or equivalently, how ``far'' an edge slide is. To make sure that $d_\mathcal{I}$ and $d_\mathcal{S}$ behave similarly with respect to scale, we let $d_\mathcal{S}$ measure the distance the sliding endpoint has traveled during an edge slide. However, this creates an interesting problem: the edge on which the endpoint is sliding may be moving and stretching/shrinking during the operation. This influences how long it takes to perform the edge slide. We need to be more specific to make this approach work: we define that (1) as the points are moving, the relative position (between $0$ and $1$ from starting endpoint to finishing endpoint) of a sliding endpoint is maintained without cost in $d_\mathcal{S}$, and (2) $d_\mathcal{S}$ measures the difference in relative position multiplied by the length $L(t)$ of the edge on which the endpoint is sliding. More tangibly, an edge slide performed by a $K$-Lipschitz stable algorithm can be performed in $t^*$ time such that $\int_0^{t^*} \frac{K}{L(t)} \mathrm{d} t = 1$, where $L(t)$ describes the length of the edge on which the endpoint slides as a function of time. Finally, the optimization function $f$ simply computes a linear interpolation of the cost on the edges of the flip graph defined by edge slides.

We now give an upper bound on $K$ below which any $K$-Lipschitz stable algorithm for kinetic EMST performs asymptotically as bad as any fixed tree. Given the complexity of the problem, our bound is fairly crude. We state it anyway to demonstrate the use of our framework, but we believe that a stronger bound exists. Before we go into the details, we first show the asymptotic approximation ratio of any spanning tree.

%\pagebreak
\begin{lemma}
	Any spanning tree on a set of $n$ points is an $O(n)$-approximation of the EMST.
\end{lemma}
\begin{proof}
	Let $T$ be an EMST on point set $P$ with total edge length $\OPT$. Additionally let $u,v\in P$, and observe that the path along $T$ from $u$ to $v$ is at least the Euclidean distance between $u$ and $v$, $\|u,v\|\leq\prop{path}_T(u,v)$. Furthermore, any path along an EMST is at most as long as the total length of an EMST, $\prop{path}_T(u,v)\leq\OPT$.
	If we now take an arbitrary spanning tree $T'$ on the same point set $P$, then we know that each edge $(u',v')$ in this spanning tree has at most length $\|u',v'\|\leq\prop{path}_T(u',v')\leq\OPT$. Since $T'$ has $n-1$ edges, its total length is $O(n)\cdot\OPT$.
\end{proof}

\begin{figure}[b]
	\centering
	\includegraphics{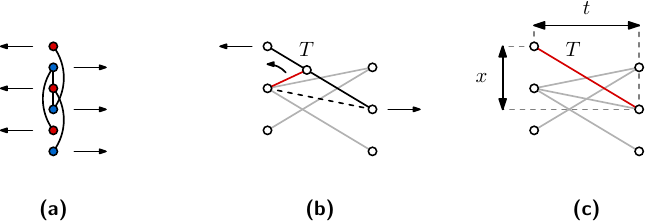}
	\caption{An instance where any $\frac{c}{\log n}$-Lipschitz stable algorithm performs poorly, for small enough constant $c>0$. \textbf{\textsf{(a)}} Red and blue points move in opposite directions. \textbf{\textsf{(b)}} The red edge slides over stretching edge $T$. \textbf{\textsf{(c)}} The length of edge $T$ at time $t$ is $\sqrt{x^2+t^2}$.}
	\label{fig:lipschitzneg}
\end{figure}

\begin{theorem}
	For a state-aware algorithm~$\mathcal{A}$ solving the kinetic EMST problem on $n$ (moving) input points, let $d_\mathcal{S}$ be the metric for edge slides, then $\LS(\text{EMST}, \frac{c}{\log n}, d_\mathcal{I}, d_\mathcal{S}) = \Omega(n)$ for a small enough constant $c > 0$.
\end{theorem}
\begin{proof}
	Consider the instance where $n$ points are placed equidistantly vertically above each other with distance $1/n$ between two consecutive points. Now let $\mathcal{A}$ be any $(c/\log n)$-Lipschitz stable algorithm for the kinetic EMST problem and let $T$ be the tree computed by $\mathcal{A}$ on this point set. We now color the points red and blue based on a $2$-coloring of $T$. We then move the red points to the left by $\frac{1}{2}$ and the blue points to the right by $\frac{1}{2}$ in the time interval $[0, 1]$ (see Figure~\ref{fig:lipschitzneg}a).
	This way every edge of $T$ will be stretched to a length of $\Omega(1)$ and thus the length of $T$ will be $\Omega(n)$. On the other hand, the length of the EMST in the final configuration is $\OPT = O(1)$. Therefore, we must perform an edge slide (see Figure~\ref{fig:lipschitzneg}b). However, we show that $\mathcal{A}$ cannot complete any edge slide. Consider any edge of $T$ and let $x$ be the initial (vertical) distance between its endpoints. Then the length of this edge can be described as $L(t) = \sqrt{x^2 + t^2}$ (see Figure~\ref{fig:lipschitzneg}c). Now assume that we want to slide an endpoint over this edge. To finish this edge slide before $t = 1$, we require that $\int_0^1 \frac{c}{\log n \sqrt{x^2 + t^2}} \mathrm{d} t \geq 1$. This solves to $c \log(1/x + \sqrt{1 + 1/x^2}) \geq \log n$. However, since $x \geq 1/n$, we get that $c \log(1/x + \sqrt{1 + 1/x^2}) \leq c \log(n + \sqrt{1 + n^2}) < \log n$ for $c$ small enough. Finally, since one edge slide can reduce the length of only one edge to $o(1)$, the cost of the solution at $t = 1$ computed by $\mathcal{A}$ is $\Omega(n)$. Thus, $\LS(\text{EMST}, \frac{c}{\log n}, d_\mathcal{I}, d_\mathcal{S}) = \Omega(n)$ for a small enough constant $c > 0$.
\end{proof}

\section{Conclusion}\label{sec:conclusion}
We presented a framework for algorithm stability, which includes three types of stability analysis: event stability, topological stability, and Lipschitz stability. The framework also distinguishes between three models for algorithms on time-varying data: stateless, state-aware and clairvoyant algorithms. We have demonstrated the use of this framework by applying the different types of analysis to algorithms for the kinetic EMST problem, deriving new results of independent interest. By illustrating different types of stability analysis with increasing degrees of complexity, we hope to make stability analysis for algorithms more accessible. In work that followed the initial conference publication of this paper, we gave further evidence of the applicablity of our framework by using it to analyze the stability of other geometric problems, namely the $k$-center problem~\cite{hoog2018topological} and orientation-based shape descriptor problems~\cite{meulemans2019stability}.

However, the framework that we presented does not give a complete picture: we do not yet consider the trade-offs between stability and running time. In this context one might ask the following questions: is there a clairvoyant algorithm to efficiently compute a stable function of solutions over time that is optimal with regard to solution quality? Or, in a more restricted sense, can one efficiently compute one solution that is best for all inputs over time? Even for state-aware algorithms, we can further study the design of efficient algorithms that are $K$-Lipschitz stable and perform well with regard to solution quality. These questions have been studied for certain geometric problems~\cite{DBLP:journals/jgaa/AkitayaBBCMSS23,meulemans2019stability,hoog2018topological}, but are open for many others.

Additionally, there are still many problems left open for the kinetic EMST problem. Future work could focus on the algorithmic models other than state-aware algorithms. However, for state-aware algorithms we could strive to tighten the bounds on the topological stability ratio for both edge slides and rotations. We would also like to extend our results on the Lipschitz stability of kinetic EMSTs, by considering $K$-Lipschitz stable solutions for different values of $K$. A description of the Lipschitz stability ratio as a function of $K$ would be the ultimate goal.

\mypar{Acknowledgments}
W. Meulemans and J. Wulms were (partially) supported by the Netherlands eScience Center (NLeSC) under grant number 027.015.G02. B. Speckmann and K. Verbeek were supported by the Netherlands Organisation for Scientific Research (NWO) under project no.~639.023.208 and no.~639.021.541, respectively.

\bibliographystyle{plainurl}
\bibliography{references}

% \begin{thebibliography}{1}

% \bibitem{ah77}
% K.~Appel and W.~Haken.
% \newblock Every planar map is four colorable.
% \newblock {\em Illinois Journal of Mathematics}, 21:439--567, 1977.

% \bibitem{h1886}
% P.~Hampson.
% \newblock {\em The Romance of Mathematics}.
% \newblock Oxford Press, Bigg City, 1886.

% \bibitem{rsst97}
% N.~Robertson, D.~P. Sanders, P.~Seymour, and R.~Thomas.
% \newblock The four-colour theorem.
% \newblock {\em Journal of Combinatorial Theory, Series B}, 70(1):2--44, 1997.

% \end{thebibliography}

\end{document}